\documentclass[11pt,letter]{article}
\usepackage{graphicx}
\usepackage{authblk}
\usepackage{auxlib}
\usepackage{verbatim}
\usepackage{times}
\usepackage{fullpage}
\usepackage[bottom=1.5in,top=1in,left=1in,right=1in]{geometry}

\title{Multiplicative assignment with upgrades}
\author[1]{Alexander Armbruster\footnote{The first author was funded by the Deutsche Forschungsgemeinschaft (DFG, German
Research Foundation) – project 551896423.}}
\author[2]{Lars Rohwedder\footnote{The second author was supported by Dutch Research Council (NWO) project
``The Twilight Zone of Efficiency: Optimality of Quasi-Polynomial Time Algorithms" [grant number OCEN.W.21.268].}}
\author[1]{Stefan Weltge}
\author[1]{Andreas Wiese}
\author[1]{Ruilong Zhang}
\affil[1]{Technical University of Munich}
\affil[2]{University of Southern Denmark}
\date{\today}

\begin{document}

\maketitle

\begin{abstract}
We study a problem related to submodular function optimization and
the exact matching problem for which we show a rather peculiar status:
its natural LP-relaxation can have fractional optimal vertices, but there is always also an optimal \emph{integral} vertex, which we can also compute in polynomial time.

More specifically, we consider the \emph{multiplicative assignment problem with upgrades} in
which we are given a set of customers and suppliers and we seek to
assign each customer to a different supplier. Each customer has a
demand and each supplier has a regular and an upgraded cost for each
unit demand provided to the respective assigned client. Our goal is
to upgrade at most $k$ suppliers and to compute an assignment in
order to minimize the total resulting cost. This can be cast as the
problem to compute an optimal matching in a bipartite graph
with the additional constraint that we must select $k$ edges from
a certain group of edges, similar to selecting $k$ red edges in the
exact matching problem. Also, selecting the suppliers to be upgraded
corresponds to maximizing a submodular set function under a cardinality
constraint.

Our result yields an efficient LP-based algorithm to solve our problem optimally.
In addition, we provide also a purely strongly polynomial-time algorithm for it.
As an application, we obtain exact algorithms
for the upgrading variant of the problem to schedule jobs on identical
or uniformly related machines in order to minimize their sum of completion
times, i.e., where we may upgrade up to $k$ jobs to reduce their
respective processing times. 
\end{abstract}

\thispagestyle{empty} \newpage{}

\setcounter{page}{1}

\section{Introduction}

Consider a supply network with some given suppliers $I$ and customers
$J$ satisfying $|I|\ge|J|$ where each customer must be assigned
to exactly one supplier and we can assign at most one customer to
each supplier. All suppliers produce the same good, each supplier
$i\in I$ at an individual \emph{cost} $c_{i}\ge0$ per unit, and
each customer $j\in J$ has a given \emph{demand} $d_{j}\ge0$. The
\emph{multiplicative assignment problem} asks for finding an assignment
$\pi:J\to I$ of minimum total cost $\sum_{j\in J}d_{j}\cdot c_{\pi(j)}$.
Finding such an optimal assignment is easy: Sort the suppliers by
cost in non-decreasing order, sort the customers by demand in non-increasing
order, and match them accordingly.

We consider the following extension of the problem, which we call
the \emph{multiplicative assignment problem with upgrades}. For each
supplier $i\in I$ we are given in addition an \emph{improved cost}
$b_{i}\in[0,c_{i}]$ that becomes effective if we decide to \emph{upgrade}
supplier $i$, e.g., modelling to improve the production processes
or facilities of $i$ to make production more cost-efficient. Given
a number $k\in\mathbb{Z}_{\ge0}$, our task is to find a subset of
$k$ suppliers to upgrade that results in an assignment of minimum
cost.

For each selected set $X\subseteq I$ of suppliers to be upgraded,
we denote by $\cost(X)$ the cost of the resulting optimal assignment.
We will show that the function $\cost(X)$ is supermodular, 
and hence
our problem is equivalent to maximizing a submodular function subject
to a cardinality constraint. While for some classes of submodular
functions this problem is known to be polynomially solvable, e.g.,
for gross substitute functions via a greedy algorithm \cite{leme2017gross}, to
the best of our knowledge $\cost(X)$ does not belong to any of them.
In particular, we show that the greedy algorithm may fail to compute
an optimal solution.

Maybe unexpectedly, a special case of our problem is the upgrade-variant
of the classical problem of scheduling jobs non-preemptively to minimize
their average completion time, either on a single machine, on identical
machines, or on uniformly related machines (see \Cref{secSchedulingUpgrades}).
Here, each job has a regular processing time and an upgraded (smaller)
processing time, and we are allowed to choose $k$ jobs to be upgraded.
This is motivated by assigning highly skilled personnel, improvement
of equipment, or by subcontracting certain tasks, and it is also referred
to as \emph{crashing }a job in the project management literature,
see e.g.~\cite{gutjahr2000stochastic}. In the scheduling literature,
this has also been studied under the notion of \emph{testing} a job
\cite{damerius2023scheduling} (often even in the stochastic setting
where the upgraded processing times are not known a priori). Recently,
Damerius, Kling, Li, Xu, and Zhang~\cite{damerius2023scheduling}
gave a PTAS for the upgrade-variant of minimizing the job's average
completion time on a single machine, but they leave open whether the
problem can be solved in polynomial time\footnote{We remark that the PTAS works even in the weighted setting, which
the authors show to be NP-hard \cite{damerius2023scheduling}.}.

We consider now the setting where $|I|=|J|$ which we may assume w.l.o.g.\ 
by adding customers with zero demand. In this case, our problem can
be also described using the notion of bipartite perfect matchings.
We introduce a node for each supplier $i\in I$ and a node for each
customer $j\in J$. For each pair $(i,j)$ there are two edges: one
\emph{blue} edge with cost $c_{i}\cdot d_{j}$ corresponding to assigning
$j$ to $i$ \emph{without} upgrading~$i$, and one \emph{red} edge
with cost $b_{i}\cdot d_{j}$ corresponding to this assignment \emph{with}
upgrading~$i$, see \Cref{fig:redblue}. Hence, an optimal solution
corresponds to a perfect matching that has exactly $k$ red edges.
Therefore, our problem is related to the exact bipartite perfect matching
problem in which we are given a bipartite graph with red and blue
edges and an integer $k$ and need to decide whether there is a perfect
matching with exactly $k$ red edges. This problem has the intriguing
status that it admits a beautiful polynomial time \emph{randomized}
algorithm due to Mulmuley, Vazirani, and Vazirani~\cite{mulmuley1987matching},
but despite several attempts, no deterministic polynomial-time algorithm
has been found so far. In contrast to the general case of that problem,
our bipartite graphs are complete. However, for us it is not sufficient
to compute \emph{any} perfect matching with exactly $k$ red edges,
but we seek such a perfect matching that optimizes our objective function. For this variant only a randomized pseudopolynomial algorithm is known, see \Cref{secSchedulingUpgrades}, which also leads to a randomized pseudopolynomial time algorithm for the multiplicative assignment problem with upgrades.

\begin{figure}[t]
\centering \includegraphics[width=0.6\linewidth]{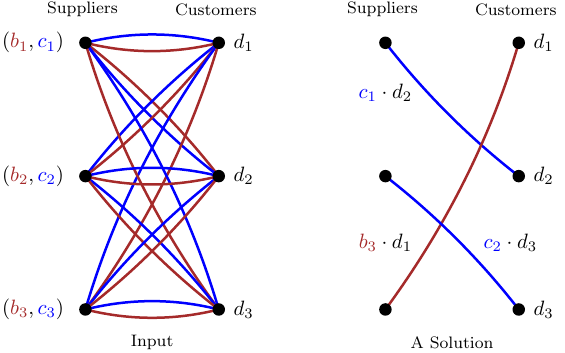}
\caption{Instance of the multiplicative assignment problem with upgrades
with $|I|=|J|=3$ and $k=1$, viewed as a perfect matching problem
in a bipartite graph with red and blue edges. The right matching corresponds
to upgrading supplier $3$ and assigning customer $1$ to supplier
$3$, customer $2$ to supplier $1$, and customer $3$ to supplier
$2$, resulting in a cost $b_{3}d_{1}+c_{1}d_{2}+c_{2}d_{3}$.}\label{fig:redblue}
\end{figure}

The usual bipartite matching problem can be formulated by a straight-forward
linear program (LP). It is well-known that its constraint matrix is
totally unimodular and, hence, all vertices of the corresponding polytope
are integral. If we require that at most $k$ red edges are chosen,
this adds ``only'' one single constraint to the LP. However, then
the resulting polytope is no longer integral. One could hope to strengthen
this LP formulation with additional constraints and variables in order
to describe exactly the convex hull of all integral points. However,
no simple characterization of the convex hull is known and
Jia, Svensson, and Yuan~\cite{jia2023exact} showed that (for some
$k$) any such extended formulation has exponential size.
This makes it difficult to apply LP methods to solve this problem.

\subsection{Our contribution}

First, we consider the natural LP formulation of the multiplicative
assignment problem. Similar to the exact bipartite matching problem,
the resulting polytope is identical to the polytope of the corresponding
bipartite matching instance with the additional constraint that at
most $k$ red edges are chosen.
We show that non-integral vertices may be even optimal LP solutions for our cost function.
However, we show that in this case there is always an optimal
\emph{integral} vertex as well! Moreover, we show how to compute such an optimal
integral vertex in polynomial time.

The fact that there is always an optimal integral
vertex while there may also exist other optimal non-integral vertices
contrasts
many other optimization problems in which either \emph{all}
(optimal) vertices of the corresponding LP are integral (e.g., shortest
path, minimum cost flow, or other problems for which the constraint
matrix is totally unimodular) or there are instances in which \emph{none}
of the optimal vertices are integral (e.g., TSP, knapsack, vertex
cover, etc.). In fact, we are not aware of any other optimization
problem with the status we prove for the multiplicative assignment
problem with upgrades.

The feasible region of our LP relaxation is the bipartite matching
polytope with one single additional constraint. Therefore, if an optimal
vertex of our LP relaxation is not integral, then it must be
a convex combination of two vertices corresponding to integral matchings.
One of these two
matchings upgrades a set $A$ of strictly \emph{less} than $k$ suppliers
and the other matching upgrades a set $B$ of strictly \emph{more} than $k$ suppliers. One key insight is to view each such supplier
$i\in A\cup B$ as the interval $[b_{i},c_{i}]$, i.e., the interval
defined by its improved cost $b_{i}$ and its usual cost $c_{i}$.
We prove that if we consider only the suppliers in which $A$ and
$B$ differ, i.e., the suppliers in the symmetric difference $A\symdiff B$,
no two of their intervals pairwise contain each other! Using this,
we derive a procedure that computes two new sets $A'$ and $B'$ that
are ``more balanced'' than $A$ and $B$: we assign the suppliers
in $A\symdiff B$ alternatingly to $A'$ and $B'$.
Then, we show that from $A,B,A'$, and $B'$ we can form a new pair having 
again the properties of $A,B$ mentioned above (in particular that a convex combination
of their corresponding optimal matchings is an optimal LP-solution)
but such that the number of upgraded
suppliers differ less than for the pair $A,B$.
We prove that this ``balancing''
process terminates with an integral optimal matching after a linear
number of iterations. If for an instance the suppliers' costs are
pairwise distinct and the customer's demands are also pairwise distinct,
we prove that then even \emph{all} optimal vertices are integral.
However, also in that case the polytope of the LP-relaxation can still
have fractional (sub-optimal) vertices since the mentioned conditions
affect only the objective function but not the polytope.

Let us denote by $h(k')$ the cost of an optimal (integral) solution
that upgrades exactly $k'$ suppliers, for each $k'\in\mathbb{N}$.
Our proof above implies that the linear interpolation of $h$ is convex.
From this property, we also derive a purely combinatorial, strongly
polynomial-time algorithm for the multiplicative assignment problem
with upgrades, which reduces the problem to at most $|I|$ instances
of weighted bipartite matching \emph{without} any additional (upgrading) constraint.
In each iteration, we start with one (integral) solution $S_{1}$
that is optimal for upgrading a certain number $k_{1}<k$ suppliers,
and another solution $S_{2}$ that is optimal for upgrading $k_{2}>k$
suppliers. Given these solutions, we solve a Lagrangian relaxation
of the multiplicative assignment problem in which the constraint for
choosing $k$ red edges is moved to the objective function, penalizing
red edges. This penalty is chosen such that $S_{1}$ and $S_{2}$
are still optimal for the resulting cost function. If the computed
solution is strictly better than $S_{1}$ and $S_{2}$ (for the new
cost function) then this yields a new integral solution $S'$ that
is optimal for updating a certain number $k'$ of suppliers with $k_{1}<k'<k_{2}$;
we replace $S_{1}$ or $S_{2}$ by $S'$ and continue with the next
iteration. Otherwise, we compute an optimal integral solution with
the same method we used above to balance the sets $A$ and $B$ iteratively.

As a consequence, we obtain an exact strongly polynomial-time algorithm
for the aforementioned upgrade-variants of scheduling jobs on $m$
identical or uniformly related parallel machines to minimize their
average completion time, which hence extends the work of~\cite{damerius2023scheduling}.
A natural question is whether our results still hold in a more general setting, for
example, when the given bipartite graph is not complete, if we can reduce the demand of some given number of customers
(similar to upgrading the costs of the suppliers), or if the suppliers are partitioned into subsets
such that we can upgrade a certain number from each of them (similar to a partition matroid). However,
for all these cases we prove that the resulting LP-relaxation might not have optimal integral vertices and, hence,
our results do not extend to them.

\subsection{Other related work}

{\bf Scheduling with Testing or Resources.}
Scheduling with testing under explorable uncertainty to bridge online and robust optimization was introduced by~\cite{durr2020adversarial} and extended to various scheduling models~\cite{albers2021explorable,bampis2021speed,gong2024approximation}. 
These problems were initially motivated by code optimization, where jobs (programs) are either executed directly or pre-processed to reduce running time. \cite{damerius2023scheduling} studied single-machine total completion time minimization in four settings (online/offline, uniform/general testing), proposing a PTAS for general testing offline but leaving open the complexity of uniform testing offline. 
Our work resolves this and extends the results to the setting of multiple machines.
Scheduling with Resource Allocation~\cite{grigoriev2007machine,kumar2009unified} dynamically adjusts job processing times via renewable resources (e.g., computational power) while respecting budget constraints.
These problems also share some similarities with our problem.
Unlike our testing-based model (fixed post-testing runtimes), resource allocation allows runtime adjustments, creating a fundamental distinction.

{\bf \noindent Budget Constrained Matching.}
Budget-constrained combinatorial optimization is well-studied~\cite{aissi2024minimum,chekuri2011multi,grandoni2014new,ravi1993many}. 
The closest work is budget-constrained bipartite matching~\cite{chekuri2011multi,grandoni2014new}, where given a bipartite graph with edges that have a weight and multiple costs, we seek a minimum weight perfect matching such that all cost dimensions satisfy the budget limit.
The problem is clearly NP-hard and generalizes our problem.
A PTAS via iterative rounding is shown in~\cite{chekuri2011multi,grandoni2014new}.

{\bf \noindent Submodular Maximization.}
Since $\cost(X)$ is supermodular (see \cref{sec:supermodularity}), our problem reduces
to maximizing a submodular function under a cardinality constraint. 
For general non-negative, monotone, submodular functions, a $(1-1/e)$-approximation is achievable, but improving this requires super-polynomially many queries~\cite{nemhauser1978best}. 
For matroid rank and gross substitutes functions, exact optimal solutions can be computed in polynomial time using the greedy algorithm~\cite{leme2017gross}. However, our problem is not
captured by these settings; indeed, we show that the greedy algorithm might fail to compute an optimal solution for our problem (see \Cref{apx:greedy}).

\section{Basic definitions and LP-formulation}

\label{secPreliminiaries}We first define the our problem formally
and introduce some basic concepts, including the LP-formulation we
will work with. We are given a set of suppliers $I$ and and a set
of customers $J$ where we assume that $|I|\ge|J|$. For each supplier
$i\in I$ we are given costs $b_{i}$ and $c_{i}$ such that $0\le b_{i}\le c_{i}$,
for each customer $j\in J$ we are given a demand $d_{j}\ge0$, and
in addition we are given a value $k\in\{0,\dots,|I|\}$. Our goal
is to compute a subset $A\subseteq I$ with $|A|\le k$ (corresponding
to the suppliers we upgrade) and a one-to-one map $\pi:J\to I$ (corresponding
of the assignment of customers to suppliers); our objective is to
minimize
\begin{equation}
\sum_{j\in J:\pi(j)\in A}b_{\pi(j)}d_j+\sum_{j\in J:\pi(j) \in I\setminus A}c_{\pi(j)}d_j.\label{eq:objective}
\end{equation}

For any fixed set $A\subseteq I$, it is easy to determine an assignment
$\pi$ that minimizes~\eqref{eq:objective}: Simply sort the suppliers
by their \emph{effective} cost ($b_{i}$ if $i\in A$ and $c_{i}$
if $i\notin A$) in non-decreasing order, sort the customers by demand
in non-increasing order, and match them accordingly. Note that this
works independently of the cardinality of $A$, so even if $|A|>k$.
Hence, we have:

\begin{restatable}{lemma}{lemComputeAssignmentOnly}
\label{lemComputeAssignmentOnly}
Given a set $A\subseteq I$,
in polynomial time we can compute a one-to-one map $\pi:J\to I$ that
minimizes $\sum_{j\in J:\pi(j)\in A}b_{\pi(j)}d_j+\sum_{j\in J:\pi(j) \in I\setminus A}c_{\pi(j)}d_j.$
\end{restatable}

For the sake of completeness, we give a formal proof in \Cref{secDeferredProofs}.
Given a set $A\subseteq I$,
we denote by $\cost(A)$ the cost of the optimal assignment according
to \Cref{lemComputeAssignmentOnly}. Hence, we want to compute
a set $A\subseteq I$ with $|A|\le k$ that minimizes $\cost(A)$.
While $\pi$ can be computed by a greedy algorithm once $A$ is fixed,
we cannot simply use a greedy algorithm to compute an optimal set $A$ (see \Cref{apx:greedy}
for a counterexample). Therefore, we need a more sophisticated approach
to solve our problem.

In the remainder of this paper, we will assume that $|I|=|J|$,
which we can ensure by adding dummy customers with zero demand.
Note that then $\pi$ is bijection.

\paragraph{Bipartite matching.}

Since $|I|=|J|$ it is convenient to study our problem using the notion
of perfect matchings. To this end, we fix $G$ to be the bipartite
(multi-)graph with vertices $V=I\dot{\cup}J$ where each pair $i\in I$
and $j\in J$ are connected by a \emph{blue }edge with cost $c_{i}\cdot d_{j}$
corresponding to assigning $j$ to $i$ \emph{without} upgrading $i$,
and a \emph{red }edge with cost $b_{i}\cdot d_{j}$ corresponding
to this assignment \emph{with} upgrading~$i$. With this terminology, our problem asks for
finding a minimum-cost perfect matching in $G$ that contains at most
$k$ red edges.

\paragraph{Linear program.}

We will study the following natural linear programming formulation
of our problem:
\begin{alignat}{3}
\min & \,\,\, & \sum\nolimits_{(i,j)\in I\times J}b_{i}d_{j}x_{i,j}+c_{i}d_{j}y_{i,j} &  & \,\,\,\,\,\, & \forall j\in J\label{LP1}\\
\mathrm{s.t.} &  & \sum\nolimits_{i\in I}x_{i,j}+y_{i,j} & =1 &  & \forall j\in J\label{LP2}\\
 &  & \sum\nolimits_{j\in J}x_{i,j}+y_{i,j} & =1 &  & \forall i\in I\label{LP3}\\
 &  & \sum\nolimits_{(i,j)\in I\times J}x_{i,j} & \le k\label{LP4}\\
 &  & x_{i,j} & \ge0 &  & \forall(i,j)\in I\times J\label{LP5}\\
 &  & y_{i,j} & \ge0 &  & \forall(i,j)\in I\times J.\label{LP6}
\end{alignat}
Setting $x_{i,j}=1$ corresponds to assigning customer $j$ to supplier
$i$ \emph{with} upgraded cost (i.e., via a red edge) and $y_{i,j}=1$
corresponds to assigning customer $j$ to supplier $i$ \emph{without}
upgrading its cost (i.e., via a blue edge). For a given instance,
denote by $P(k)$ the feasible region of our LP. Without Constraint~\eqref{LP4}, by the Birkhoff–von Neumann theorem (see, e.g.~\cite[Thm.~18.1]{schrijver2003combinatorial}), this set would be identical to the perfect matching polytope of
$G$, i.e., the convex hull of all integral vectors $(x,y)$ corresponding
to perfect matchings in $G$, and thus each vertex would be integral.

However, due to the additional Constraint~\eqref{LP4} the polytope $P(k)$
may have a fractional vertex that might even be optimal for our
objective function. Consider for example the instance defined by $I=J=\{1,2\}$,
$(b_{1},c_{1})=(0,1)$, $(b_{2},c_{2})=(2,3)$, $d_{1}=d_{2}=1$ and
$k=1$, see \Cref{fig:fractional-opt}.
The resulting polytope $P(1)$ has the fractional vertex $(x^{*},y^{*})$
defined by $x_{1,2}^{*}=x_{2,1}^{*}=y_{1,1}^{*}=y_{2,2}^{*}=\frac{1}{2}$
and setting all other variables to 0.
The objective function value
of $(x^{*},y^{*})$ is $3$.
This is optimal since every point $(x,y)$ in $P(1)$ satisfies
\begin{align*}
\sum\nolimits_{(i,j)\in I\times J}b_{i}d_{j}x_{i,j}+c_{i}d_{j}y_{i,j} & =2x_{2,1}+2x_{2,2}+y_{1,1}+y_{1,2}+3y_{2,1}+3y_{2,2}\\
 & =3(\underbrace{x_{1,1}+x_{2,1}+y_{1,1}+y_{2,1}}_{=1})+3(\underbrace{x_{1,2}+x_{2,2}+y_{1,2}+y_{2,2}}_{=1})\\
 & \,\,\,\,\,-2(\underbrace{x_{1,1}+x_{1,2}+y_{1,1}+y_{1,2}}_{\le1})-(\underbrace{x_{1,1}+x_{1,2}+x_{2,1}+x_{2,2}}_{\le1})\\
 & \geq 3.
\end{align*}
On the other hand there is also the integral vertex $(\bar{x},\bar{y})$
defined by $\bar{x}_{1,1}=1$, $\bar{y}_{2,2}=1$, and setting all
other variables to 0, whose objective function value is also 3.
Hence, for this instance the polytope $P(1)$ has an optimal vertex that
is integral. In the next section, we show that this is true for \emph{any}
instance, even if the corresponding polytope $P(k)$ has \emph{also} fractional optimal vertices.

\begin{figure}
    \centering
    \includegraphics[width=0.8\linewidth]{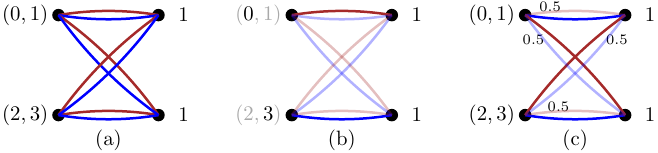}
    \caption{Instance with an optimal integral vertex (middle) and an optimal fractional vertex (right).}
    \label{fig:fractional-opt}
\end{figure}

\section{Integral optimal vertices}

\label{secIntegrality}

In this section, we prove that for any instance of our problem, the corresponding
polytope $P(k)$ admits an optimal vertex that is integral. Our proof
is constructive, yielding a polynomial-time algorithm that computes
an optimal integral vertex.

\begin{theorem} \label{thmMain} For any instance of the multiplicative
assignment problem, the corresponding polytope $P(k)$ admits an optimal
vertex that is integral. Moreover, we can compute such a vertex in
polynomial time.
\end{theorem}

Notice that the definition of $P(k)$ is independent of the costs and demands.
Since in general $P(k)$ does not coincide with the convex hull of its integer points, it may still have fractional vertices that are optimal with respect to other objective functions (not corresponding to our problem).
We also provide a sufficient conditions under which \emph{all optimal}
vertex solutions are integral.

\begin{theorem}\label{thm:AllOptimalVerticesAreInteger} If in an
instance of the multiplicative assignment problem all values $(b_{i})_{i\in I}$
and $(c_{i})_{i\in I}$ are pairwise different and also all values
$(d_{j})_{j\in J}$ are pairwise different, then every optimal vertex
of $P(k)$ is integral.\end{theorem}

We first focus on the proof of \Cref{thmMain}; the proof of \Cref{thm:AllOptimalVerticesAreInteger}
is given in \Cref{secAllIntegral}.
Assume we are given an instance of the multiplicative assignment problem.
First, we compute a vertex solution $(\bx,\by)$ of $P(k)$. If $(\bx,\by)$
is integral, we are done. Assume that this is not the case. Let $P$
denote the polytope we obtain by omitting the constraint $\sum\nolimits_{(i,j)\in I\times J}x_{i,j}\le k$
from our LP, i.e., omitting Constraint~\eqref{LP4}. Hence, $P$
is the convex hull of all (integral) vectors corresponding to perfect matchings. Therefore, $(\bx,\by)$
is \emph{not} a vertex of $P$. However, since $P(k)$ can be obtained from $P$ by adding only one
inequality, one can show that for $(\bx,\by)$ this inequality is tight and that
$(\bx,\by)$ lies
on an edge of $P$. Thus, it is a convex combination of two
vertices of $P$ which correspond to two perfect matchings, one upgrading less than $k$ suppliers and one upgrading more than $k$ suppliers. We compute
these points and matchings using the following lemma. Formally, for
any set \emph{$X\subseteq I$ }we denote by $M_{X}$ the optimal perfect
matching that \emph{upgrades a set $X\subseteq I$, }i.e.,\emph{ }such
that $X$ equals the set of suppliers that are incident to some red
edge in $M$ (if there is more than one such optimal matching we break ties arbitrarily). Also, for any perfect matching $M$ we denote
by $\chi(M)$ its corresponding vertex of $P$.

\begin{restatable}{lemma}{lemVertexEdgeMatchings}
\label{lemVertexEdgeMatchings}
Given a vertex $(\bx,\by)$ of $P(k)$, in polynomial
time we can compute two sets $A,B\subseteq I$ with $|A|<k<|B|$, their corresponding optimal
perfect matchings $M_{A}$ and $M_{B}$, and a value $\lambda\in(0,1)$
such that $(\bx,\by)=\lambda\cdot\chi(M_{A})+(1-\lambda)\cdot\chi(M_{B})$.
\end{restatable}
A proof of above lemma is given in \Cref{secDeferredProofs}.
Note that any pair $A', B' \subseteq I$ with $|A'| < k < |B'|$ has a corresponding convex combination 
that upgrades exactly $k$ supplies. We denote
by $f_{A',B'}(k)$ the objective value of this 
fractional solution. We call $A', B'$ an \emph{optimal pair} if this solution is an optimal fractional solution, see formal definition below.

\begin{definition} Let $A',B'\subseteq I$. The pair $(A',B')$ is
an \emph{optimal pair} if $|A'|<k<|B'|$ and the optimal solution
value of the LP \eqref{LP1}-\eqref{LP6} equals
\[
f_{A',B'}(k):=\tfrac{|B'|-k}{|B'|-|A'|}\cost(A')+\tfrac{k-|A'|}{|B'|-|A'|}\cost(B').
\]
\end{definition}

To gain some intuition, if $(\bx,\by)$ upgrades (fractionally) exactly
$k$ suppliers, then $(\bx,\by)$ is the mentioned convex combination
by construction and it has an objective function value of $f_{A,B}(k)$.
Also, then $f_{A,B}(k)$ equals the optimal solution value of our LP since
$(\bx,\by)$ is an optimal LP-solution; hence $(A,B)$ is an optimal pair.
\begin{lemma}The pair $(A,B)$ is an optimal pair.

\end{lemma}

\begin{proof} Let $\alpha$ denote the cost of $M_{A}$ and $\beta$
denote the cost of $M_{B}$. Since
$(\bx,\by)=\tfrac{|B|-k}{|B|-|A|}\chi(M_{A})+\tfrac{k-|A|}{|B|-|A|}\chi(M_{B})$
the optimum value of the LP is indeed equal to
\[
\tfrac{|B|-k}{|B|-|A|}\alpha+\tfrac{k-|A|}{|B|-|A|}\beta\ge\tfrac{|B|-k}{|B|-|A|}\cost(A)+\tfrac{k-|A|}{|B|-|A|}\cost(B)=f_{A,B}(k). \qedhere
\]
\end{proof}

It may happen that in $M_{A}$ (or in $M_{B}$) some (upgraded) supplier
in $A$ (or in $B$) is assigned to a customer with zero demand, contributing
a cost of zero to the objective function value. Intuitively, this
wastes the upgrade of this supplier since it would contribute a cost
of zero in the objective also if it were not upgraded. To avoid certain
technical complications, we would like to remove such suppliers from
$A$ and $B$. Formally, we would like $A$ and $B$ to be \emph{simple}, where we define a set $X\subseteq I$ to be \emph{simple} if in $M_{X}$
no supplier is assigned to a customer with zero demand. We can easily
make $A$ and $B$ simple by just removing iteratively suppliers that
are assigned to customers with zero demand. Formally, we apply the following
lemma to our set $(A,B)$
obtain a new simple optimal pair or even
directly an optimal integral vertex (in which case we are done).

\begin{restatable}{lemma}{lemSimplify}
\label{lemSimplify}
Let $(\hat A,\hat B)$ be an optimal pair.
In polynomial time we can compute sets $\hat A'\subseteq \hat A$ and $\hat B'\subseteq \hat B$
that are both simple such that $(\hat A',\hat B')$ is an optimal pair or $\chi(M_{\hat B'})$
is an optimal vertex.
\end{restatable}
\begin{proof}
    By iteratively removing suppliers that are assigned to customers with zero demand, we obtain simple sets $\hat A' \subseteq \hat A$ and $\hat B' \subseteq \hat B$ with $\cost(\hat A') = \cost(\hat A)$ and $\cost(\hat B') = \cost(\hat B)$.
	Let $\lambda = \frac{k - |\hat A'|}{|\hat B'| - |\hat A'|}$.
    Since $|\hat A'| < k$, the optimum LP value is at most $\cost(\hat A')$ and hence
    \[
        \cost(\hat A') \ge f_{\hat A',\hat B'}(k) = (1-\lambda) \cost(\hat A') + \lambda \cost(\hat B'),
    \]
    which yields $\cost(\hat A') \ge \cost(\hat B')$.
    If $|\hat B'| \ge k$, we thus have
    \begin{align}
        f_{\hat A',\hat B'}(k) & = \cost(\hat B') + \tfrac{|\hat B'| - k}{|\hat B'| - |\hat A'|} (\cost(\hat A') - \cost(\hat B')) \nonumber \\
        & = \cost(\hat B) + \tfrac{|\hat B'| - k}{|\hat B'| - |\hat A'|} (\cost(\hat A) - \cost(\hat B)) \nonumber \\
        & \le \cost(\hat B) + \tfrac{|\hat B'| - k}{|\hat B'| - |\hat A|} (\cost(\hat A) - \cost(\hat B)) \nonumber \\
        & = \cost(\hat A) + \tfrac{k - |\hat A|}{|\hat B'| - |\hat A|} (\cost(\hat B) - \cost(\hat A)) \nonumber \\
        & \le \cost(\hat A) + \tfrac{k - |\hat A|}{|\hat B| - |\hat A|} (\cost(\hat B) - \cost(\hat A)) = f_{\hat A,\hat B}(k).\label{eq:simpleLastStep}
    \end{align}
    If $|\hat B'| > k$, this implies that $(\hat A',\hat B')$ is also optimal.
    Otherwise, let $\hat B' \subseteq \hat B'' \subseteq \hat B$ such that $|\hat B''| = k$.
    Since $\cost(\hat B'') = \cost(\hat B)$, we may replace $\hat B'$ by $\hat B''$ in the above inequalities to obtain
    \[
        f_{\hat A,\hat B}(k) \ge f_{\hat A,\hat B''}(k) = \cost(\hat B'') = \cost(\hat B'),
    \]
    which matches the second case of the claim.
\end{proof}

Due to \Cref{lemSimplify} we can assume that our optimal pair $(A,B)$
is simple. We will show next that this implies that $(A,B)$ is \emph{clean},
which we define as follows.

\begin{definition} Let $A',B'\subseteq I$. The pair $(A',B')$ is
\emph{clean }if the symmetric difference $A'\symdiff B'=(A'\setminus B')\cup(B'\setminus A')$
does not contain two suppliers $i,i'$ with $b_{i}<b_{i'}$ and $c_{i'}<c_{i}$.

\end{definition}

Thinking of each supplier $i$ as the interval $[b_{i},c_{i}]$, this
means that a pair $(A',B')$ is clean if $A'\symdiff B'$ does not
contain an interval that is contained in the interior of
another interval from $A'\symdiff B'$.

\begin{restatable}{lemma}{lemNice}
\label{lemNice}
If $(A',B')$ is an optimal pair
where both $A',B'$ are simple, then $(A',B')$ is clean. 
\end{restatable}
\begin{proof} 
Suppose for the sake of contradiction that there are
$i,i'\in A\symdiff B$ with $b_{i'}<b_{i}$ and $c_{i}<c_{i'}$.
Intuitively, this leads to a contradiction because there is an optimal fractional solution that selects both $i$ and $i'$ non-integrally, yet upgrading $i'$ is strictly better than upgrading $i$. Thus, by an exchange argument the supposedly optimal fractional solution can be improved.
Below, we formalize this intuition.

As
$i\in A\symdiff B$, there is a blue edge in $(M_{A}\cup M_{B})$
connecting $i$ to some $j\in J$. As $i'\in A\symdiff B$, there
is red edge in $(M_{A}\cup M_{B})$ connecting $i'$ to
some $j'\in J$. Since $A$ and $B$ are simple, we have that
\begin{equation}
d_{j'}>0\label{eqw47g34dfsfsd}
\end{equation}
holds. Letting $(\bx,\by)=\tfrac{|B|-k}{|B|-|A|}\chi(M_{A})+\tfrac{k-|A|}{|B|-|A|}\chi(M_{B})$,
notice that $y_{i,j},x_{i',j'}>0$ and hence $x_{i,j'},y_{i',j}<1$.
Thus, there exists some $\varepsilon>0$ such that all entries of
$(\tilde{\bx},\tilde{\by})$ defined via
\[
\tilde{x}_{e}\coloneqq\begin{cases}
x_{e}+\varepsilon & \text{if }e=(i,j')\\
x_{e}-\varepsilon & \text{if }e=(i',j')\\
x_{e} & \text{else}
\end{cases}\qquad\text{and}\qquad\tilde{y}_{e}\coloneqq\begin{cases}
y_{e}+\varepsilon & \text{if }e=(i',j)\\
y_{e}-\varepsilon & \text{if }e=(i,j)\\
y_{e} & \text{else}
\end{cases}
\]
are in $[0,1]$. Notice that $(\tilde{\bx},\tilde{\by})$ is a valid
LP solution. The difference of the objective values of $(\tilde{\bx},\tilde{\by})$
and $(\bx,\by)$ is equal to
\[
\varepsilon\cdot(b_{i}d_{j'}+c_{i'}d_{j}-c_{i}d_{j}-b_{i'}d_{j'})<\varepsilon\cdot(b_{i'}d_{j'}+c_{i}d_{j}-c_{i}d_{j}-b_{i'}d_{j'})=0,
\]
where the strict inequality follows from~\eqref{eqw47g34dfsfsd}.
This means that $(\bx,\by)$ is not an optimal LP solution, a contradiction
to the optimality of $A,B$. \end{proof}

Thus, from now on we may assume that $(A,B)$ is a clean and optimal
pair, and we will maintain this property when we adjust $A$ and $B$
in the following.
The fact that $(A,B)$ is clean allows us to order the elements (interpreted
as intervals) in $A\symdiff B$: By relabeling the suppliers, we may
assume that $A\symdiff B=\{1,\dots,\ell\}\eqqcolon[\ell]$ where $b_{1}\le b_{2}\le\dots\le b_{\ell}$
and $c_{1}\le c_{2}\le\dots\le c_{\ell}$.
Given this ordering, we define
two new sets $A'$, $B'$ that both contain $A\cap B$, and we redistribute the suppliers in $A\symdiff B$
alternatively. Formally, we define $A'\coloneqq(A\cap B)\cup\{2s:s\in[\lfloor\ell/2\rfloor]\}$
and $B'\coloneqq(A\cap B)\cup\{2s-1:s\in[\lceil\ell/2\rceil]\}$.
An example is shown in \Cref{fig:redistribution}.
We call $(A',B')$ a \emph{redistribution} of $(A,B)$. Our main technical
contribution is the following lemma which we will prove in \Cref{subsec:Proof-redistribution-lemma}.

\begin{figure}
    \centering
    \includegraphics[width=0.8\linewidth]{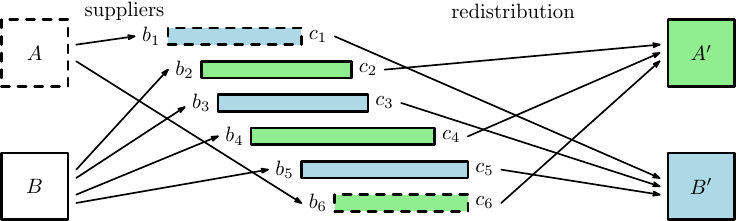}
    \caption{Illustration for the redistribution. The figure shows only the suppliers in $A\symdiff B$.
    The dotted and solid rectangles are suppliers in $A$ and $B$, respectively. The green and blue rectangles are suppliers in $A'$ and $B'$ after redistribution.}
    \label{fig:redistribution}
\end{figure}

\begin{lemma}[Redistribution lemma]\label{lem:ConvexityAfterRearrangment}
Let $(A',B')$ be a redistribution of a clean optimal pair $(A,B)$.
We have $|A|<|A'|\le|B'|<|B|$ and $\cost(A')+\cost(B')\le\cost(A)+\cost(B)$.
\end{lemma}

\begin{restatable}{corollary}{corRedistributionNiceOptimal}
\label{corRedistributionNiceOptimal}
Let $(A',B')$
be a redistribution of a clean optimal pair $(A,B)$. Then one of the
following holds:
\begin{itemize}
\item[a)] $\chi(M_{A'})$ or $\chi(M_{B'})$ is an optimal (integral) vertex
or
\item[b)] one of $(A,B')$, $(A,A')$, $(A',B)$, $(B',B)$ is a clean optimal
pair.
\end{itemize}
\end{restatable}
\begin{proof}
    Let $s \coloneqq |B| - |A| > 0$, $t \coloneqq \cost(A) - \cost(B)$, and $r \coloneqq |B| \cost(A) - |A| \cost(B)$.
    By the construction of $A'$ and $B'$, we have $|A'| + |B'| = |A| + |B|$.
    Moreover, by \Cref{lem:ConvexityAfterRearrangment}, we know that $\cost(A') + \cost(B') \le \cost(A) + \cost(B)$ holds.
    We obtain
    \[
        s \cdot (\cost(A') + \cost(B')) + t \cdot (|A'| + |B'|)
        \le s \cdot (\cost(A) + \cost(B)) + t \cdot (|A| + |B|) = 2r.
    \]
    This means that some $X \in \{A',B'\}$ must satisfy $s \cdot \cost(X) + t \cdot |X| \le r$.
    If $|X| = k$, then this yields $\cost(X) \le \frac{r - tk}{s} = f_{A,B}(k)$, as claimed in a).
    If $|X| < k < |B|$, then
    \begin{align*}
        f_{X,B}(k)
        & = \cost(B) + \tfrac{|B| - k}{|B| - |X|} (\cost(X) - \cost(B)) \\
        & \le \cost(B) + \tfrac{|B| - k}{|B| - |X|} (\tfrac{r - tk}{s} - \cost(B)) \\
        & = \cost(B) + \tfrac{|B| - k}{|B| - |A|} \tfrac{|B| - k}{|B| - |X|} (\cost(A) - \cost(B)) \\
        & \le \cost(B) + \tfrac{|B| - k}{|B| - |A|} (\cost(A) - \cost(B)) \\
        & = f_{A,B}(k),
    \end{align*}
    and hence $(X,B)$ is optimal.
    If $|A| < k < |X|$, we analogously obtain $f_{A,X}(k) \le f_{A,B}(k)$, and hence $(A,X)$ is optimal.
    Notice that $A \symdiff X \subseteq A \symdiff B$ and $X \symdiff B \subseteq A \symdiff B$, and hence both $(A,X)$ and $(X,B)$ are clean.
    So, in both cases we obtain b).
\end{proof}

Note that in a) we are done (we can check this by comparing the corresponding
costs to the cost of $(\bx,\by)$). Otherwise, we replace $(A,B)$
by the new clean optimal pair that we denote by $(\hat{A},\hat{B})$.
Since $|A|<|A'|\le|B'|<|B|$ the difference $|\hat{B}|-|\hat{A}|$
is strictly smaller than the difference $|B|-|A|$. Thus, we can conclude:

\begin{lemma}\label{lem:one-iteration}Given a clean optimal pair $(A,B)$,
in polynomial time we can compute an optimal integral vertex of $P(k)$
or a clean optimal pair $(\hat{A},\hat{B})$ such that $|\hat{B}|-|\hat{A}|<|B|-|A|$.
\end{lemma}
Given the initial optimal pair $(A,B)$, we apply \Cref{lem:one-iteration}
at most $|I|$ times and eventually obtain an optimal integral vertex
of $P(k)$. This completes the proof of \Cref{thmMain}.

Above, we solved the LP to compute the fractional vertex $(\bx,\by)$
and, based on it, the initial optimal pair. In \Cref{secCombinatorial}, we
will see that we can instead also compute the initial optimal pair
by a combinatorial algorithm, which yields a purely combinatorial algorithm
to solve the multiplicative assignment problem.

\subsection{Proof of the redistribution lemma}\label{subsec:Proof-redistribution-lemma}

\newcommand{\cp}{\Pi}

The bounds on $|A'|, |B'|$ follow immediately from the fact that the redistribution guarantees $|A'| \le |B'| \le |A'| + 1$ and that $|A| < k < |B|$.

We start by giving some intuition for the comparison of costs. Instead of directly comparing costs we can equivalently compare the change in cost, i.e., it suffices to prove that
\begin{equation}
    (\cost(A') - \cost(\emptyset)) + (\cost(B') - \cost(\emptyset))
    \le (\cost(A) - \cost(\emptyset)) + (\cost(B) - \cost(\emptyset)).\label{eq:costs}
\end{equation}
The difficulty comes from the complicated structure of the $\cost(\cdot)$ function.
Even the change in cost for upgrading a single supplier $i$, i.e., $\cost(\{i\}) - \cost(\emptyset)$, is not obvious since the optimal matching may be different depending on whether we upgrade~$i$. On the other hand, in the 
special case where the open interval $(b_i, c_i)$ does not contain any of the cost values $b_{i'}, c_{i'}$, $i'\in I$, the order of suppliers by cost does not change if we upgrade $i$ and therefore the optimal matching also does not change. Thus, in this case we have the simple identity
$\cost(\{i\}) - \cost(\emptyset) = (b_i - c_i) \cdot d_j$ where $j\in J$ is the customer that supplier $i$ is matched to.
Our proof strategy is to view the cost changes
$\cost(A) - \cost(\emptyset)$, $\cost(B) - \cost(\emptyset)$, 
etc.\ as
the sum of step-wise upgrades where each step does not change the optimal matching. Within each step we can then more easily compare the change in costs between the different sets. 

Formally,
for each $t \ge 0$, we will consider a ``truncated'' instance $\cp_t$ in which costs cannot get upgraded below $t$.
More precisely, $\cp_t$ has the same suppliers $I$, customers $J$, standard supplier costs $c$, and customer demands $d$.
For a supplier $i \in I$, the upgraded cost in $\cp_t$ will be equal to $\min\{t,c_i\}$ if $b_i \le t$, and remains $b_i$ otherwise.
For suppliers $S \subseteq I$ to upgrade, we define the resulting cost in $\cp_t$ as $\cost(S,t)$.

Let $t_1 < t_2 < \dots < t_\ell$ with $\{t_1,\dots,t_{\ell}\} = \{b_i : i\in I\} \cup \{c_i : i\in I\}$.
We will show that for every $h \in [\ell-1]$ it holds that
\begin{multline}
    \label{eqsdgh83h4g74}
    \cost(A',t_h) - \cost(A',t_{h+1}) + \cost(B',t_h) - \cost(B',t_{h+1}) \\
    \le \cost(A,t_h) - \cost(A,t_{h+1}) + \cost(B,t_h) - \cost(B,t_{h+1})
\end{multline}
Summing over all $h \in [\ell - 1]$, we obtain Inequality \eqref{eq:costs}
which yields
the lemma's statement.

Let us fix $h \in [\ell - 1]$.
We may assume that $J = [n]$ and $d_1 \ge d_2 \ge \dots \ge d_n$.
Consider a 
subset $S \subseteq I$ 
for which we would like to determine $\cost(S,t_h) - \cost(S,t_{h+1})$. 
To this end, note that for a supplier $i \in I$ we can only have $c_i \le t_h$, $b_i \le t_h < t_{h+1} \le c_i$, or $t_{h+1} \le b_i$.
We denote by $I_{\le}$ the suppliers of the first type, i.e., $I_{\le} \coloneqq \{ i \in I : c_i \le t_h \}$.
By $S_*$ we denote the suppliers of the second type that are also contained in $S$, i.e., $S_* \coloneqq \{ i \in S : b_i \le t_h, \, t_{h+1} \le c_i \}$. By $S_{\ge} \coloneqq I \setminus (I_{\le} \cup S_*)$ we denote all remaining suppliers.
Note that $S_{\ge}$ depends on $S$ but is not necessarily a subset of $S$.

Assuming that we only upgrade the suppliers in $S$, 
for $t \in \{t_h,t_{h+1}\}$
let us consider the effective costs
of each supplier $i\in I$ in the instance $\cp_t$, i.e., $\max \{b_i, \min\{t, c_i\}\}$
if $i\in S$ and $c_i$ if $i \not \in S$.
The effective cost of any supplier in $I_{\le}$ is at most $t_h$ and independent of whether $t=t_h$ or $t=t_{h+1}$.
Similarly, the effective cost of any supplier in $S_{\ge}$ is at least $t_{h+1}$ and again independent of whether $t=t_h$ or $t=t_{h+1}$.
The effective cost of any supplier $i \in S_*$ is equal to $t$.
In $\cp_t$, an optimal perfect matching that upgrades exactly the suppliers in $S$ can be constructed by
assigning the suppliers in $I_{\le}$ to the first $|I_{\le}|$ customers, the suppliers in $S_{\ge}$ to the last $|S_{\ge}|$ customers, and the suppliers in $S_*$ to the remaining customers.
Note that we may use the same assignment for $t=t_h$ and $t=t_{h+1}$.
The costs of two matchings differ by
\begin{equation}
    \label{eq438gh9g99h}
    \cost(S,t_h) - \cost(S,t_{h+1}) = (t_h-t_{h+1}) \sum_{j = |I_{\le}| + 1}^{|I_{\le}| + |S_*|} d_j.
\end{equation}
We claim that
\begin{equation}
    \label{eq483gj94gj}
    \sum_{j = |I_{\le}| + 1}^{|I_{\le}| + |A'_*|} d_j
    +
    \sum_{j = |I_{\le}| + 1}^{|I_{\le}| + |B'_*|} d_j
    \ge
    \sum_{j = |I_{\le}| + 1}^{|I_{\le}| + |A_*|} d_j
    +
    \sum_{j = |I_{\le}| + 1}^{|I_{\le}| + |B_*|} d_j.
\end{equation}
holds.
Notice that multiplying this inequality by $t_h - t_{h+1} < 0$ and using~\eqref{eq438gh9g99h} for the sets $A, B, A'$, $B'$, we obtain~\eqref{eqsdgh83h4g74} and are done.

To show \eqref{eq483gj94gj}, let $p \coloneqq \min \{|A_*|, |B_*|\}$, $q \coloneqq \max \{|A_*|, |B_*|\}$, $p' \coloneqq \min \{|A'_*|, |B'_*|\}$, and $q' \coloneqq \max \{|A'_*|, |B'_*|\}$.
Recall that by construction $A \cap B = A' \cap B'$ and $A \symdiff B = A' \symdiff B'$. 
In particular, we have $s \coloneqq p + q = p' + q'$.
Recall also that since the pair $(A,B)$ is clean, the suppliers $A \symdiff B$ can be ordered such that $b_i$ and $c_i$ are both non-decreasing.
In this ordering, the suppliers in $(A_* \cup B_*)\cap (A \symdiff B)$ are consecutive.
Recall the suppliers in $A \symdiff B$ have been alternatingly redistributed to $A'$ and $B'$.
This implies that $q' \le p' + 1$ and thus $p' = \lfloor s/2 \rfloor$ and $q' = \lceil s/2 \rceil $.
Letting $r \coloneqq |I_\le|$ and $g \coloneqq \lceil s/2 \rceil - p$, we obtain~\eqref{eq483gj94gj} via
\begin{alignat*}{10}
    \sum_{j = |I_{\le}| + 1}^{|I_{\le}| + |A'_*|} d_j + \sum_{j = |I_{\le}| + 1}^{|I_{\le}| + |B'_*|} d_j
    & = \sum_{j = 1}^{p'} d_{j+r} && + \sum_{j = 1}^{q'} d_{j+r} \\
    & = \sum_{j = 1}^{\lceil s/2 \rceil} d_{j+r} && + \sum_{i = 1}^{\lfloor s/2 \rfloor} d_{j+r} \\
    & = 2 \sum_{j = 1}^{p} d_{j+r} && + \sum_{j = p + 1}^{\lfloor s/2 \rfloor} d_{j+r} && + \sum_{j = p + 1}^{\lceil s/2 \rceil} d_{j+r} \\
    & \ge 2 \sum_{j = 1}^{p} d_{j+r} && + \sum_{j = p + 1}^{\lfloor s/2 \rfloor} d_{j+r+g} && + \sum_{j = p + 1}^{\lceil s/2 \rceil} d_{j+r} \\
    & = 2 \sum_{j = 1}^{p} d_{j+r} && + \sum_{j = \lceil s/2 \rceil + 1}^{q} d_{j+r} && + \sum_{j = p + 1}^{\lceil s/2 \rceil} d_{j+r} \\
    & = \sum_{j = 1}^{p} d_{j+r} && + \sum_{j = 1}^{q} d_{j+r} \\
    & = \sum_{j = |I_{\le}| + 1}^{|I_{\le}| + |A_*|} d_j && + \sum_{j = |I_{\le}| + 1}^{|I_{\le}| + |B_*|} d_j.
\end{alignat*}

\subsection{Sufficient condition for integrality of all optimal vertices}
\label{secAllIntegral}

This section aims to prove \Cref{thm:AllOptimalVerticesAreInteger}.
We consider the instance that satisfies (\rom{1}) all $b_i,c_i$ are pairwise different and (\rom{2}) all $d_j$ are pairwise different, and show that every optimal vertex solution is integral.
We remark that both these two conditions are required, and missing one of the conditions will lead to a fractional vertex.
In \Cref{secPreliminiaries}, we showed that the instance only satisfies (\rom{1}) has fractional vertices.
Below we give another example, which shows that the instance that only satisfies (\rom{2}) also has fractional vertices.

\paragraph*{Fractional Vertices without (\rom{1}).}
Consider the following instance defined by $I=J=\set{1,2,3}$, where $(b_1,c_1)=(0,1),(b_2,c_2)=(1,1),(b_3,c_3)=(1,4)$, and $d_1=3,d_2=2,d_3=1$, and $k=1$.
See \Cref{fig:fractional-opt-2}(a) for an example.
One can easily check that the optimal integral solution is $6$ by either upgrading node $1$ or $3$; see \Cref{fig:fractional-opt-2}(b).
There exists a fractional solution that is also equal to $6$; see \Cref{fig:fractional-opt-2}(c).
Using a similar calculation in \Cref{secPreliminiaries}, one can show that any feasible solution to LP has a value of at least $6$.
Thus, the fractional solution shown in \Cref{fig:fractional-opt-2}(c) is an optimal fractional solution. As its support is a single cycle, it is on an edge of $P$. As the two endpoints of this edge upgrade $0$ and $2$ suppliers, this fractional solution is also a vertex of $P(1)$.

\begin{figure}
    \centering
    \includegraphics[width=0.8\linewidth]{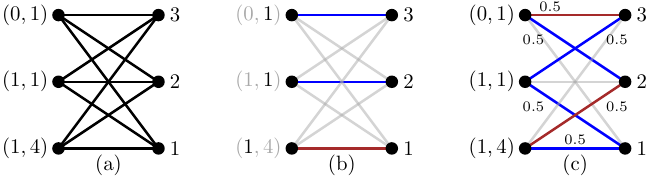}
    \caption{Illustration of the fractional vertices without (\rom{1}).}
    \label{fig:fractional-opt-2}
\end{figure}
\bigskip

We now prove \Cref{thm:AllOptimalVerticesAreInteger}.
Let $\bx$ be an optimal vertex of linear-program~\eqref{LP1}--\eqref{LP4}, which lies on the edge of the perfect matching polytope $P$.
Thus, $\bx$ is a convex combination of two matchings $M_A$ and $M_B$, which upgrades a set $A\subseteq I$ and $B\subseteq I$ with $|A|<k<|B|$, respectively.
We slightly abuse the notations and also use $A$ and $B$ to represent red edges in $M_A$ and $M_B$.
Similar to \Cref{secIntegrality}, we assume $A\symdiff B = \{1,2,\dotsc,\ell\}$ and create sets $A' \coloneqq (A \cap B) \cup \{ 2s-1 : s \in [\lceil \ell/2 \rceil] \}$ and $B' \coloneqq (A \cap B) \cup \{ 2s : s \in [\lfloor \ell/2 \rfloor] \}$.
Before showing \Cref{thm:AllOptimalVerticesAreInteger}, we first show that the inequality in \Cref{lem:ConvexityAfterRearrangment} is not strict; this will be useful later.

\begin{lemma}
We have $\cost(A')+\cost(B')=\cost(A)+\cost(B)$. \label{lem:InequalityNotStrict}
\end{lemma}
\begin{proof}
Recall that \Cref{lem:ConvexityAfterRearrangment}
already proves
$\cost(A')+\cost(B') \le \cost(A)+\cost(B)$.
We show that if the inequality is strict, then it shall contradict the optimality of $\bx$.
Note that $\bx_{|A'|}=\frac{|B|-|A'|}{|B|-|A|}\chi(M_A)+\frac{|A'|-|A|}{|B|-|A|}\chi(M_B)$ and $ \bx_{|B'|}=\frac{|B|-|B'|}{|B|-|A|}\chi(M_A)+\frac{|B'|-|A|}{|B|-|A|}\chi(M_B)$ are LP solutions for upgrading $|A'|$ and $|B'|$ suppliers, respectively.
Suppose for the sake of contradiction that $\cost(A')+\cost(B')<\cost(A)+\cost(B)$. 
Note that $c(\bx_{|A'|})+c(\bx_{|B'|})=\cost(A)+\cost(B)$, where $c(\bx)$ is the LP's objective value when the solution is $\bx$. 
Thus, for $S=A'$ or $S=B'$, we obtain $\cost(S)< c(x_{|S|})$. 
If $|S|\geq k$ then $\frac{k-|A|}{|S|-|A|}\chi(M_A)+\frac{|S|-k}{|S|-|A|}\chi(M_S)$ is an LP solution for upgrading $k$ suppliers with cost strictly less than $c(\bx)$. 
Similarly if $|S|<k$ then $\frac{k-|S|}{|B|-|S|}\chi(M_S)+\frac{|B|-k}{|B|-|S|}\chi(M_B)$ is an LP solution for upgrading $k$ suppliers with cost strictly less than $c(\bx)$. 
This is a contradiction as $\bx$ is an optimal LP solution. Thus $\cost(A')+\cost(B')= \cost(A)+\cost(B)$.    
\end{proof}
\begin{lemma}\label{Lemt}
If all demands are pairwise different then for every $t\in\R\setminus\{b_{i},c_{i}:i\in I\}$ it holds that
\begin{equation*}
||\{i\in A:b_{i}<t<c_{i}\}|-|\{i\in B:b_{i}<t<c_{i}\}|| \le 1 .
\end{equation*}
\end{lemma}
\begin{proof}
    Recall the proof of the proof of \Cref{lem:ConvexityAfterRearrangment}.
    We observe that none of the inequalities~\eqref{eq483gj94gj} can be strict since otherwise we obtain $\cost(A')+\cost(B')<\cost(A)+\cost(B)$, contradicting \Cref{lem:InequalityNotStrict}.

	For the sake of contradiction, suppose that there is some $t\in\R\setminus\{b_{i},c_{i}:i\in I\}$ such that $||\{i\in A:b_{i}<t<c_{i}\}|-|\{i\in B:b_{i}<t<c_{i}\}||>1$. Then we have $||A_*| - |B_*|| > 1$ 
    for the interval $(t_h,t_{h+1})$ containing~$t$. This implies $p\leq q-2$ and therefore also $p+1\leq \lfloor s/2 \rfloor$ and $g\geq 1$. As the demands are pairwise different and thus strictly decreasing, we obtain $\sum_{j=p+1}^{\lfloor s/2 \rfloor}d_{j+r}>\sum_{j=p+1}^{\lfloor s/2 \rfloor}d_{j+r+g}$, which directly makes~\eqref{eq483gj94gj} strict.
\end{proof}

To show \Cref{thm:AllOptimalVerticesAreInteger}, it suffices to prove the following lemma.

\begin{lemma}\label{lem:DifferenceOfOnePerCycle}
    Let $C$ be a cycle in $M_A \symdiff M_B$. Then $||A\cap C|-|B\cap C||\leq 1$.
\end{lemma}
Using this Lemma, it is straightforward to show that every optimal vertex is integral.
\begin{proof}[Proof of \Cref{thm:AllOptimalVerticesAreInteger} assuming \Cref{lem:DifferenceOfOnePerCycle}]
    Suppose that $\bx$ is a vertex of polytope $P(k)$ and not integral.
    As $\bx$ lies on an edge of the perfect matching polytope $P$ connecting $M_A$ and $M_B$, the edges $M_A \symdiff M_B$ form one unique cycle.
    By \Cref{lem:DifferenceOfOnePerCycle} this implies $||A|-|B||\leq 1$. If $|A|=|B|$ then $|A|=|B|=k$ implying both $\chi(M_A)$ and $\chi(M_B)$ are points in polytope $P(k)$.
    Thus, $\bx$ is not a vertex of $P(k)$. 
    If $||A|-|B||=1$, this contradicts the fact that the values $|A|, |B|$ and $k$ are all integers since $\min\{|A|,|B|\}<k<\max \{|A|,|B|\}=\min\{|A|,|B|\}+1$. 
\end{proof}
Thus it remains to prove \Cref{lem:DifferenceOfOnePerCycle}. 
To this end, we show that when $M_A$ and $M_B$ satisfy some properties, the cycle in $M_A\symdiff M_B$ shall fulfill \Cref{lem:DifferenceOfOnePerCycle}.
These properties are stated in  \Cref{lem:DifferenceOfOnePerCycleInductionHypothesis}.

\begin{lemma}\label{lem:DifferenceOfOnePerCycleInductionHypothesis}
    Let $\bar{G}=(\bar{V}\cup [\bar{n}], \bar{E})$ with $\bar{n}=|\bar{V}|$ being a complete bipartite Graph. Let $\bar{A}, \bar{B}\subseteq \bar{V}$ be two subsets of nodes and let $M_{\bar{A}}, M_{\bar{B}} : [\bar n] \rightarrow [\bar n]$ be two perfect matchings with the following properties:
    \begin{enumerate}[label=(\ref{lem:DifferenceOfOnePerCycleInductionHypothesis}\ablue{\alph*}),leftmargin=*,align=left]
        \item For $v\in \bar{V} \setminus (\bar{A}\symdiff\bar{B})$ we have $|M_{\bar{A}}(v)-M_{\bar{B}}(v)|\leq 1$.
        \label{prop:induction:other-nodes}
        \item For $v\in \bar{A}\setminus \bar{B}$ we have $M_{\bar{A}}(v)\leq M_{\bar{B}}(v)$ and for $v\in \bar{B}\setminus \bar{A}$ we have $M_{\bar{B}}(v)\leq M_{\bar{A}}(v)$.
        \label{prop:induction:base-case}
        \item There does not exist $v_1, v_2 \in A \symdiff B$ such that $\min\{M_{\bar{A}}(v_1), M_{\bar{B}}(v_1)\}<\min\{M_{\bar{A}}(v_2), M_{\bar{B}}(v_2)\}$ and  $\max\{M_{\bar{A}}(v_1), M_{\bar{B}}(v_1)\}>\max\{M_{\bar{A}}(v_2), M_{\bar{B}}(v_2)\}$.
        \label{prop:induction:including}
    \end{enumerate}
    Then for every cycle $C$ in $M_{\bar{A}}\symdiff M_{\bar{B}}$ we have $||\bar{A}\cap C|-|\bar{B}\cap C||\leq 1$.
\end{lemma}
\begin{proof}[Proof of \Cref{lem:DifferenceOfOnePerCycle} assuming \Cref{lem:DifferenceOfOnePerCycleInductionHypothesis}.]
    We show that $M_A\symdiff M_B$ fulfills the requirements of \Cref{lem:DifferenceOfOnePerCycleInductionHypothesis}.

    \paragraph{Property \ref{prop:induction:other-nodes}.}
    To show that for each $v\in V \setminus (A\symdiff B)$ we have $|M_{A}(v)-M_{B}(v)|\leq 1$, 
    we distinguish two cases: (\Rom{1}) $v\in V\setminus (A\cup B)$ and (\Rom{2}) $v\in A\cap B$.
    
    For Case~(\Rom{1}), let $v\in V \setminus (A\cup  B)$ and let $t=c_v-\varepsilon$ with $\varepsilon$ small enough such that there exists no $c_{v'}$ or $b_{v'}$ in $(t, c_{v})$.
    By \Cref{Lemt}, we obtain $||\{v'\in A:b_{v'}<t<c_{v'}\}|-|\{v'\in B:b_{v'}<t<c_{v'}\}||\leq 1$.
    Recall that in the matchings $M_A$ and $M_B$, the suppliers are sorted by their effective cost in increasing order, and the customers are sorted by decreasing order of demands and then matched accordingly. This implies that for $S\in \{A, B\}$ we have
    \begin{align*}
        M_{S}(v) &=|\{v'\in V: c_{v'}\leq c_v\}|+|\{v'\in S:b_{v'}<c_v<c_{v'}\}| \\
        &=|\{v'\in V: c_{v'}\leq c_v\}|+|\{v'\in S:b_{v'}<t<c_{v'}\}| .     
    \end{align*}
    As $||\{v'\in A:b_{v'}<t<c_{v'}\}|-|\{v'\in B:b_{v'}<t<c_{v'}\}||\leq 1$ we also obtain $|M_{A}(v)-M_{B}(v)|\leq 1$.
    
    In Case (\Rom{2}), let $v\in A\cap B$. The argument is almost the same. Let $t=b_v+\varepsilon$ with $\varepsilon$ be small enough so that there exist no $c_{v'}$ or $b_{v'}$ in $(b_{v}, t+\varepsilon)$.  
    Then by the same argument as before
    \begin{align*}
    M_{S}(v) &= |\{v'\in V: c_{v'}\leq b_v\}|+|\{v'\in S:b_{v'}\leq b_v<c_{v'}\}| \\
    &=|\{v'\in V: c_{v'}\leq b_v\}|+|\{v'\in S:b_{v'}<t<c_{v'}\}| . 
    \end{align*}
    This then implies $|M_{A}(v)-M_{B}(v)|\leq 1$ using \Cref{Lemt}.

    \paragraph{Property \ref{prop:induction:base-case}.}
    Let $v\in A\setminus B$. The case $v\in B \setminus A$ is symmetric. 

    First we argue that $A$ and $B$ are simple. As all $d_j$ are pairwise different, there is at most one customer $\bar{j}$ with demand $0$. If there is none, it is clear that $A$ and $B$ are simple. 
    So suppose there is exactly one. 
    Then an optimal assignment upgrading $S\subseteq I$ assigns $\bar{j}$ to an upgraded supplier only if all other suppliers are also upgraded. 
    This implies $S=I$ and is thus only possible for $B=I$. 
    As $|A|<k<|B|$ implies $|A|\leq n-2$, there is a supplier $i\not\in A$ assigned to a customer with positive demand. 
    Recall the proof of \cref{lemSimplify}. 
    As Upgrading $i$ improves the cost we have $\cost(A)>\cost(A\cup \{i\})\geq \cost(B)$ and thus $f_{A, B}(k)<\cost(B)$. 
    So $\chi(M_{B'})$ is not an optimal vertex. 
    As $\cost(B)>\cost(A)$ and $|B'|<|B|$ the inequality \eqref{eq:simpleLastStep} is strict. 
    This implies $f_{A, B}(k)>f_{A', B'}(k)$, which is a contradiction to the fact that $(A, B)$ is an optimal pair. 
    So we obtain that, indeed, $A$ and $B$ are simple.
  
    As in the previous cases $M_{A}(v)=|\{v'\in V: c_{v'}\leq b_v\}|+|\{v'\in A:b_{v'}\leq b_v<c_{v'}\}|$. 
    Note that $(A,B)$ is a simple optimal pair and by \Cref{lemNice}, we have that $A, B$ is clean, i.e., for $v'\in A\setminus B$ with $b_{v'}\leq b_v$ that $c_{v'}\leq c_{v}$. Thus 
    \begin{align*}
        M_A(v)&\leq |\{v'\in V: c_{v'}\leq c_v\}|+|\{v'\in A\cap B:b_{v'}\leq b_v<c_v<c_{v'}\}|\\
        &\leq |\{v'\in V: c_{v'}\leq c_v\}|+|\{v'\in B:b_{v'}<c_v<c_{v'}\}|=M_B(v)
    \end{align*}

    \paragraph{Property \ref{prop:induction:including}.}
    There will be two cases: Case (\Rom{1}) $v_1,v_2\in A\setminus B$, the case $v_1,v_2\in B\setminus A$ is analogue; Case (\Rom{2}) $v_1\in A\setminus B$ and $v_2\in B\setminus A$, the case $v_2\in A\setminus B$ and $v_1\in B\setminus A$ is symmetric.
    
    In Case (\Rom{1}), we have $v_1, v_2\in A \setminus B$. 
    Then, by \ref{prop:induction:other-nodes} of \Cref{lem:DifferenceOfOnePerCycleInductionHypothesis}, we have $M_A(v_1)\leq M_B(v_1)$ and $M_A(v_2)\leq M_B(v_2)$. 
    If $c_{v_1}<c_{v_2}$ then $M_B(v_1)<M_B(v_2)$, contradicting \ref{prop:induction:base-case} of \Cref{lem:DifferenceOfOnePerCycleInductionHypothesis}.
    If $c_{v_1}>c_{v_2}$ then \Cref{lemNice} implies $b_{v_1}>b_{v_2}$ and therefore $M_A(v_1)>M_A(v_2)$, contradicting again \ref{prop:induction:other-nodes} of \Cref{lem:DifferenceOfOnePerCycleInductionHypothesis}.
    
    In Case (\Rom{2}), we have $v_1\in A\setminus B$ and $v_2\in B\setminus A$.
    Then, by \ref{prop:induction:base-case} of \Cref{lem:DifferenceOfOnePerCycleInductionHypothesis}, we have $M_A(v_1)\leq M_B(v_1)$ and $M_B(v_2)\leq M_A(v_2)$. So to prove the lemma, we need to show that $M_A(v_1)\geq M_B(v_2)$ or $M_B(v_1)\leq M_A(v_2)$. 
    
    First suppose that $ b_{v_1}<b_{v_2}$. 
    This implies $c_{v_1}<c_{v_2}$ by \Cref{lemNice}. 
    Let $t=c_{v_1}+\varepsilon$ with $\varepsilon$ be small enough so that there is neither $b_{v}\in (c_{v_1}, t)$ nor $c_{v}\in (c_{v_1}, t)$ for any $v\in V$. Then we have
    \begin{align*}
        M_{B}(v_1)&=|\{v\in V: c_v\leq c_{v_1}\}|+|\{v\in B: b_v< c_{v_1}<c_{v}\}|\\
        &=|\{v\in V: c_v\leq t\}|+|\{v\in B: b_v< t<c_{v}\}|\\
        &\leq |\{v\in V: c_v\leq t\}|+|\{v\in A: b_v< t<c_{v}\}|+1
    \end{align*}
        Note that for $v\in A$ with $b_v< t<c_{v}$ we either have $c_v<c_{v_2}$ or $c_v>c_{v_2}$ as $v_2\not\in A$, thus
    \begin{align*}
        M_{B}(v_1)&\leq |\{v\in V: c_v< c_{v_2}\}|+|\{v\in A: b_v< c_{v_2}<c_{v}\}|+1\\
        &\leq |\{v\in V: c_v\leq c_{v_2}\}|+|\{v\in A: b_v< c_{v_2}<c_{v}\}|\\
        &=M_{A}(v_2)
    \end{align*}

Now suppose that $ b_{v_1}>b_{v_2}$. Let $t=b_{v_2}+\varepsilon$ with $\varepsilon$ be small enough so that there is neither $b_{v}\in (b_{v_2}, t)$ nor $c_{v}\in (b_{v_2}, t)$ for any $v\in V$. Then we have
 \begin{align*}
    M_{B}(v_2)&=|\{v\in V: c_v< b_{v_2}\}|+|\{v\in B: b_v\leq b_{v_2}<c_{v}\}|\\
    &=|\{v\in V: c_v\leq t\}|+|\{v\in B: b_v< t<c_{v}\}|\\
    &\leq |\{v\in V: c_v\leq t\}|+|\{v\in A: b_v< t<c_{v}\}|+1\\
    &\leq |\{v\in V: c_v< b_{v_1}\}|+|\{v\in A: b_v< b_{v_1}<c_{v}\}|+1\\
    &\leq |\{v\in V: c_v< b_{v_1}\}|+|\{v\in A: b_v\leq b_{v_1}<c_{v}\}|\\
    &=M_{A}(v_1)
\end{align*}
\end{proof}
The last step is to show \Cref{lem:DifferenceOfOnePerCycleInductionHypothesis}.
\begin{proof}[Proof of \Cref{lem:DifferenceOfOnePerCycleInductionHypothesis}]
Given a complete bipartite graph $\bar{G}$, we consider two node sets $\bar{A},\bar{B}$ and two perfect matchings $M_{\bar{A}},M_{\bar{B}}$ that satisfy all properties stated in \Cref{lem:DifferenceOfOnePerCycleInductionHypothesis}.
The bipartite multigraph $M_{\bar{A}}\symdiff M_{\bar{B}}$ may contain more than one cycle, in which we shall consider each cycle separately and show that each cycle $C$ on its own satisfies $||\bar{A}\cap C|-|\bar{B}\cap C||\leq 1$.
This shall prove the lemma.

We prove the lemma by iteratively {\em contracting} the graph $M_{\bar{A}}\symdiff M_{\bar{B}}$.
In each contracting step, we shall delete some edges and contract two nodes, and thus we obtain a smaller graph.
Thanks to the properties stated in \Cref{lem:DifferenceOfOnePerCycleInductionHypothesis}, each contracting step ensures that (\rom{1}) either exactly one edge from both $\bar{A}$ and $\bar{B}$ is deleted or no edges from $\bar{A}\cup\bar{B}$ are removed; (\rom{2}) the new graph that results still satisfies all properties in \Cref{lem:DifferenceOfOnePerCycleInductionHypothesis}.
Finally, we reach the base case where each cycle includes at most $4$ edges and then stop contracting.
In the following, we first show that the base case always satisfies the desired property and describe the contracting step.
This finishes the proof of \Cref{lem:DifferenceOfOnePerCycleInductionHypothesis}.

\paragraph{Base Case.}
Consider any cycle $C$ in $M_{\bar{A}}\symdiff M_{\bar{B}}$ with at most $4$ edges, we claim that $\abs{\abs{\bar{A}\cap C}-\abs{\bar{B}\cap C}}\leq 1$.
We can assume that $C$ has exactly $4$ edges; otherwise, the claim is trivially valid.
We can also assume that $C$ includes exactly two nodes from $\bar{A}\cup\bar{B}$; otherwise, the claim also holds trivially.
It is easy to see that the claim is valid if $C$ contains exactly one node from $\bar{A}$ and one node from $\bar{B}$.
Thus, both nodes come from the same set.
Without loss of generality, suppose that $C$ includes two nodes in $\bar{A}$.
However, this contradicts \ref{prop:induction:base-case} of \Cref{lem:DifferenceOfOnePerCycleInductionHypothesis}.
Hence, the base case is true.

\paragraph{Contracting.}
In the first phase of contracting, we aim to contract nodes in $\bar{V}\setminus(\bar{A}\symdiff\bar{B})$.
For a node $v$ in this set, the contracting operation consists of two steps: (\rom{1}) remove two edges that are adjacent to $v$; (\rom{2}) merge the node $M_{\bar{A}}(v)$ and $M_{\bar{B}}(v)$. 
An example is shown in the cycle $C_1$ in \Cref{fig:contracting}.
By \ref{prop:induction:other-nodes} of \Cref{lem:DifferenceOfOnePerCycleInductionHypothesis}, we have $|M_{\bar{A}}(v)-M_{\bar{B}}(v)|\leq 1$; so, the new graph that results also fulfills the requirements of this lemma.
If $v\in \bar{V}\setminus (\bar{A}\cup\bar{B})$, no edge is deleted from $\bar{A}\cup\bar{B}$ being removed in this step.
If $v\in\bar{A}\cap \bar{B}$, then exactly one edge is removed from both $\bar{A}$ and $\bar{B}$.

At the end of the first phase, the graph does not include any nodes from $\bar{V}\setminus(\bar{A}\symdiff\bar{B})$.
Thus, each node of $\bar{V}$ in the new graph either belongs to $\bar{A}$ and $\bar{B}$.
In the second phase, we first find a minimal node $m\in[\bar{n}]$ such that $M_{\bar{A}}(M_{\bar{B}}(m))<m$ and $M_{\bar{B}}(M_{\bar{A}}(m))<m$. 
Such a node must exist as this is true for $m=\bar{n}$. 
The crucial property of such a node is that it satisfies 
\begin{equation}
|M_{\bar{A}}(M_{\bar{B}}(m))-M_{\bar{B}}(M_{\bar{A}}(m))|\leq 1.
\label{equ:contracting:key}
\end{equation}
\Cref{equ:contracting:key} allows us to contract $M_{\bar{A}}(M_{\bar{B}}(m))$ and $M_{\bar{B}}(M_{\bar{A}}(m))$ while the graph that results still fulfills the requirements of the lemma. 
In particular, the contracting operation consists of the following two steps: (\rom{1}) delete all edges that are adjacent to $M_{\bar{B}}(m)$ and $M_{\bar{A}}(m)$; (\rom{2}) merge the node $M_{\bar{A}}(M_{\bar{B}}(m))$ and $M_{\bar{B}}(M_{\bar{A}}(m))$.
An example is shown in the cycle $C_2$ in \Cref{fig:contracting}.
This operation removes exactly one edge from both $\bar{A}$ and $\bar{B}$.

    \begin{figure}[tb]
        \centering
        \includegraphics[width=0.7\linewidth]{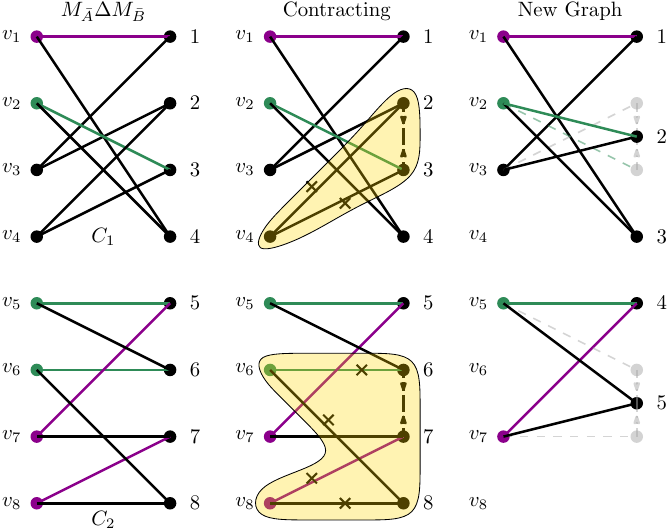}
        \caption{Illustration of Contracting Operation. The left part is the graph $M_{\bar{A}}\Delta M_{\bar{B}}$, which consists of two cycles. The purple and green nodes/edges are the set $\bar{A}$ and $\bar{B}$, respectively. The middle part shows the contracting operation, in which we shall remove some edges and contract two nodes. The edges labeled by ``$\times$'' are edges that will be deleted. The right part is the new graph after the contracting operation. The first cycle still includes more than $4$ edges, so one more contracting operation will be performed.}
        \label{fig:contracting}
    \end{figure}

Now, it remains to show \Cref{equ:contracting:key}.
For the sake of the contradiction, suppose that $|M_{\bar{A}}(M_{\bar{B}}(m))-M_{\bar{B}}(M_{\bar{A}}(m))|>1$. 
Let $n'\in [\bar{n}]$ such that $M_{\bar{A}}(M_{\bar{B}}(m))<n'<M_{\bar{B}}(M_{\bar{A}}(m))$ or $M_{\bar{A}}(M_{\bar{B}}(m))>n'>M_{\bar{B}}(M_{\bar{A}}(m))$. 
As $n'<m$ and $m$ is the minimal node for which $M_{\bar{A}}(M_{\bar{B}}(m))<m$ and $M_{\bar{B}}(M_{\bar{A}}(m))<m$, we have that $M_{\bar{A}}(M_{\bar{B}}(n'))>n'$ or $M_{\bar{B}}(M_{\bar{A}}(n'))>n'$. 
Without loss of generality, we assume that the first is the case and let $m':=M_{\bar{A}}(M_{\bar{B}}(n'))$. 
Note that $m'\neq m$ as $m$ is already matched with other nodes. 
In Case (\Rom{1}), we suppose that $m'<m$.
Recall that we have $M_{\bar{A}}(M_{\bar{B}}(m))<n'$ or $M_{\bar{B}}(M_{\bar{A}}(m))<n'$.
This is a contradiction to \ref{prop:induction:including} of \Cref{lem:DifferenceOfOnePerCycleInductionHypothesis} on $\bar{G}$ for $v_1=M_{\bar{B}}(m)$ or $v_1=M_{\bar{A}}(m)$ and $v_2=M_{\bar{B}}(n')$.
Now we consider Case (\Rom{2}) and suppose that $m'>m$, which is similar as Case (\Rom{1}).
Recall that we have $M_{\bar{A}}(M_{\bar{B}}(m))>n'$ or $M_{\bar{B}}(M_{\bar{A}}(m))>n'$.
This is a contradiction to \ref{prop:induction:including} of \Cref{lem:DifferenceOfOnePerCycleInductionHypothesis} on $\bar{G}$ for $v_1=M_{\bar{B}}(m)$ or $v_1=M_{\bar{A}}(m)$ and $v_2=M_{\bar{B}}(n')$.
Thus $|M_{\bar{A}}(M_{\bar{B}}(m))-M_{\bar{B}}(M_{\bar{A}}(m))|\leq 1$.
An example can be found in \Cref{fig:contracting:key}.
    \begin{figure}[tb]
        \centering
        \includegraphics[width=0.4\linewidth]{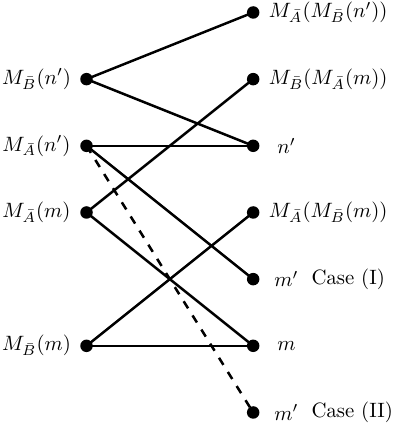}
        \caption{Illustration of the proof of \Cref{equ:contracting:key}. The minimal node $m$ such that $M_{\bar{A}}(M_{\bar{B}}(m))<m$ and $M_{\bar{B}}(M_{\bar{A}}(m))<m$ must satisfy \Cref{equ:contracting:key}; otherwise, we can always find a pair of nodes in $\bar{V}$ such that they violates \ref{prop:induction:including} of \Cref{lem:DifferenceOfOnePerCycleInductionHypothesis}.}
        \label{fig:contracting:key}
    \end{figure}
\end{proof}

\section{Combinatorial algorithm}

\label{secCombinatorial} In \Cref{secIntegrality} we have seen
that, given an optimal pair $(A,B)$, we can compute an optimal solution
by a purely combinatorial algorithm. In this section, we present a
combinatorial algorithm for computing such an optimal pair, so in
particular \emph{without }solving a linear program. This yields a
completely combinatorial algorithm for the multiplicative assignment
problem with upgrades.

Our algorithm is based on a key structural insight into the problem.
Given an instance, we consider the function $h:\{0,\ldots,|I|\}\rightarrow\mathbb{R}_{\ge0}$
such that for each $k'$ the value $h(k')$ equals the cost of the
optimal solution if we were allowed to upgrade $k'$ suppliers,
i.e., $h(k')=\min\{\cost(X):X\subseteq I,\,|X|=k'\}$. Note that
$h(k')$ is defined even if $k'>k$ where $k$ is the number of suppliers
we are allowed to upgrade in the given instance. Clearly, $h$ is
non-increasing since allowing more upgrades cannot increase the cost
of the optimal solution. A direct consequence of \Cref{thmMain} is that $h$ is also convex.

\begin{lemma}\label{lem8928916861}For any instance of the multiplicative
assignment problem with upgrades, the (linear interpolation of the) function $h$ is convex.

\end{lemma}
\begin{proof}
We have to show that $h(k_{B}) \le \frac{k_{C}-k_{B}}{k_{C}-k_{A}}h(k_{A})+\frac{k_{B}-k_{A}}{k_{C}-k_{A}}h(k_{C})$ holds for all integers $0 \le k_A < k_B < k_C \le |I|$.
Let $A,C \subseteq I$ be such that $k_A = |A|$, $k_C = |C|$, $\cost(A) = h(k_A)$, and $\cost(C) = h(k_C)$.
Let $\beta$ denote the optimum solution value of our LP for upgrading up to $k_B$ suppliers.
Recall that $\beta \le f_{A,C}(k_B)$.
Moreover, by \Cref{thmMain}, we have $h(k_B) = \beta$.
Thus, we obtain $h(k_B) \le f_{A,C}(k_B) = \frac{k_{C}-k_{B}}{k_{C}-k_{A}}h(k_{A})+\frac{k_{B}-k_{A}}{k_{C}-k_{A}}h(k_{C})$.
\end{proof}
In our algorithm, we start with a \emph{weakly }optimal pair $(A,B)$
which we define to be a pair $A,B\subseteq I$ such that $\cost(A)=h(|A|)$
and $\cost(B)=h(|B|)$, i.e., among all sets $S\subseteq I$ with
$|S|=|A|$ the set $A$ has the smallest cost, and among all sets
$S'\subseteq I$ with $|S'|=|B|$ the set $B$ has the optimal cost.
We start with the pair $(\emptyset,I)$ which is clearly weakly optimal
since $\emptyset$ and $I$ are the only subsets of $I$ with 0 and
$|I|$ elements, respectively.
Suppose we are given a weakly optimal pair $(A,B)$. We describe a
routine that asserts that $(A,B)$ is even an optimal pair, or directly
outputs an optimal integral solution, or computes another weakly optimal
pair $(A',B')$ with $|B'|-|A'|<|B|-|A|$. Thus, if we iterate this
routine for at most $|I|$ iterations, we eventually find an optimal
pair.
We define an auxiliary cost
function where for each subset $X\subseteq I$ we set
$
g_{A,B}(X)\coloneqq\cost(X)+\frac{\cost(A)-\cost(B)}{|B|-|A|}\cdot|X|.
$
Intuitively, this cost function is a type of Lagrangian relaxation
of our actual cost function where we replace Constraint~\eqref{LP4}
by the penalty term $\frac{\cost(A)-\cost(B)}{|B|-|A|}\cdot|X|$ in
the objective.
\begin{observation}
We have that $g_{A,B}(A)=g_{A,B}(B)$.
\end{observation}

To gain some intuition, let us consider a plot of the (unknown) function
$h$ and a line $\ell$ that contains the points $(|A|,\cost(A))$
and $(|B|,\cost(B))$, see \Cref{fig:combin-alg}. Intuitively, $\ell$ contains
all points with the same objective function value according to $g_{A,B}$.
Since $(A,B)$ is weakly optimal, we have that $h(|A|)=\cost(A)$
and $h(|B|)=\cost(B)$. Moreover, since $h$ is convex, one can show
that for each value $k'$ with $k'\le|A|$ or $k'\ge|B|$ the point
$(k',h(k'))$ lies on or above $\ell$. However, there might be a
value $k^{*}$ with $|A|<k^{*}<|B|$ such that $(k^{*},h(k^{*}))$
lies below $\ell$. If this is the case, then there must be a corresponding
set $X^{*}\subseteq I$ with $|X^{*}|=k^{*}$, $h(k^{*})=\cost(X^{*})$,
and $g_{A,B}(X^{*})<g_{A,B}(A)=g_{A,B}(B)$; we call such a set $X^{*}$
an \emph{extreme} set. It turns
out that if there is no extreme set $X^{*}$, then we are already done.

\begin{restatable}{lemma}{lemcheckXstar}
\label{lem:checkXstar}If there is no extreme set $X^{*}\subseteq I$,
then $(A,B)$ is an optimal pair.
\end{restatable}
\begin{proof}
    Let $S \subseteq I$ with $|S| = k$ be any optimal set.
    Since there is no extreme set, we have $g(A) \le g(S)$.
    By \Cref{thmMain}, the optimum value of the LP is equal to
    \begin{align*}
        \cost(S)
        = g(S) - \frac{\cost(A) - \cost(B)}{|B| - |A|} \cdot k
        & \ge g(X^*) - \frac{\cost(A) - \cost(B)}{|B| - |A|} \cdot k \\
        & \ge g(A) - \frac{\cost(A) - \cost(B)}{|B| - |A|} \cdot k \\
        & = f_{A,B}(k). \qedhere
    \end{align*}
\end{proof}

\begin{figure}
    \centering
    \includegraphics[width=0.95\linewidth]{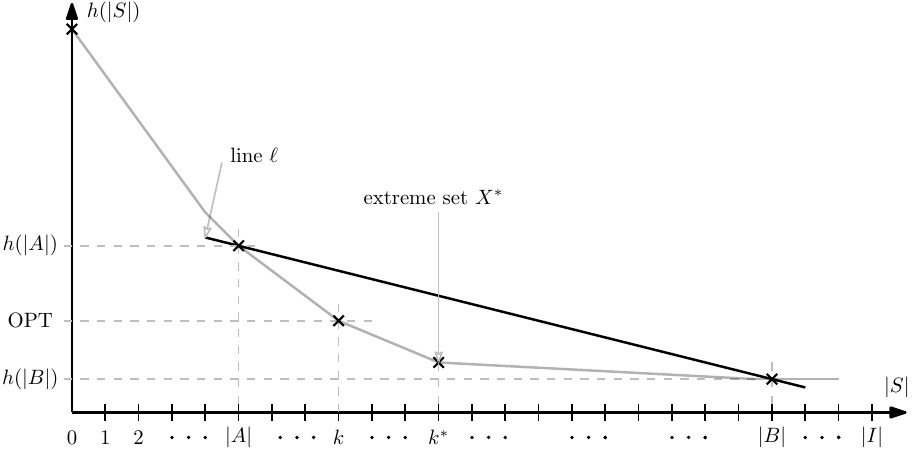}
    \caption{Illustration of the combinatorial algorithm.}
    \label{fig:combin-alg}
\end{figure}

In fact, we can easily check whether an extreme set $X^{*}$ exists.
More formally, we even compute a set $X^{*}$ that optimizes $g_{A,B}(X)$
over all sets $X\subseteq I$. All we need to do for this is to solve
an auxiliary instance of weighted bipartite matching. In this instance,
we take the bipartite graph corresponding to our instance as defined
in \Cref{secPreliminiaries} and decrease the costs of each red
edge by $\frac{\cost(A)-\cost(B)}{|B|-|A|}$. Then, we compute a minimum
cost perfect matching for the instance, i.e., \emph{without }a constraint
bounding the number of selected red edges.
\begin{lemma}
    \label{lem4hg893hg9g}
In polynomial time we can compute a set $X^{*}\subseteq I$ such that
$g_{A,B}(X^{*})\le g_{A,B}(X)$ holds for each set $X\subseteq I$. If $g_{A,B}(X^{*})<g_{A,B}(A)$,
then $X^{*}$ is an extreme set; otherwise, no extreme set exists.
\end{lemma}

Thus, if there is no extreme set, then $(A,B)$ is an optimal pair by \Cref{lem:checkXstar}.
Otherwise, let $X^*$ the set due to \Cref{lem4hg893hg9g}.
Note that, once we fix the cardinality of $X$, the
value of $g(X)$ depends only on $\cost(X)$.
This yields the following observation.
\begin{observation}
\label{obs:weakly-optimal}For every $X\subseteq I$ with $|X^{*}|=|X|$ we have $\cost(X^{*})\le\cost(X)$.
\end{observation}

Recall that since $X^*$ is extreme, we have $|A| < |X^*| < |B|$.
If $|X^{*}|=k$, then we are done since the matching $M_{X^{*}}$ is
optimal by \Cref{obs:weakly-optimal}. If $k<X^{*}$ then the
pair $(A,X^{*})$ is a weakly optimal pair with $|X^{*}|-|A|<|B|-|A|$.
Similarly, if $k>X^{*}$ then the pair $(X^{*},B)$ is weakly optimal
with $|B|-|X^{*}|<|B|-|A|$.

Thus, if we repeat this procedure for at most $|I|$ iteration, we
eventually find an optimal pair $(A,B)$. Then, we apply \Cref{lemSimplify}
to transform $(A,B)$ to a simple optimal pair, and finally we apply
\Cref{lem:one-iteration} iteratively until we found an optimal
integral solution.

\section{Applications and extensions}
\label{secApplications}
In this section, we first present details on the connection
of multiplicative assignment problems to scheduling.
We then explore natural extensions, for which there are randomized
pseudopolynomial time algorithms via simple reductions to exact
matching and for which the computational complexity with binary encoded costs is open.
Finally, we give counter-examples demonstrating that
the integrality properties we prove in \Cref{secIntegrality}
do not generalize to them.

\subsection{Scheduling with upgrades}
\label{secSchedulingUpgrades}
Consider the following scheduling problem: we are given
$n$ jobs with processing times $p_1,p_2,\dotsc,p_n$ as well
as upgraded (lower) processing times $q_1,q_2,\dotsc,q_n$ and a limit $k\in\N$ on the number of jobs to be upgraded. Furthermore, there are
$m$ machines with speeds $s_1,s_2,\dotsc,s_m$.
We have to assign each job to machines, upgrade $k$ of them,
and then process them in some order on each machine.
Our goal is to minimize the average completion time of the
jobs.
The variant without upgrades is well known to be reducible
to minimum weight perfect matching, see e.g.~\cite[Chapter~7]{elements-of-scheduling}.
While we present the setting above in the most general form,
we note that solving the problem on a single machine
does not seem obvious either. In fact, we are not aware of another
successful approach than the one presented in this paper.

To model this problem as the multiplicative assignment problem, each job is a supplier.
For each machine, we create $n$ ordered slots,
which form the $n\cdot m$ customers over all machines.
The slots correspond to positions in the schedule on a machine.
More precisely, if a job is assigned to the first slot of a machine,
this means that it is scheduled last; if it is assigned
to the second slot, then it is scheduled second-to-last, etc.
The demand for a job is $p_j$, and the upgraded demand is $q_j$.
The cost of the $\ell$th slot of machine $i$
is $\ell / (n s_i)$. Here, the rationale is that a job $j$ placed
on the $\ell$th slot will delay $\ell$ many jobs, including
itself, by $p_j / s_i$ (or $q_j / s_i$ if upgraded).
The division of $n$ is to transform the total completion time into
average completion time.
An example is shown in \Cref{fig:scheduling-reduction}.

\begin{figure}[htb]
    \centering
    \includegraphics[width=0.3\linewidth]{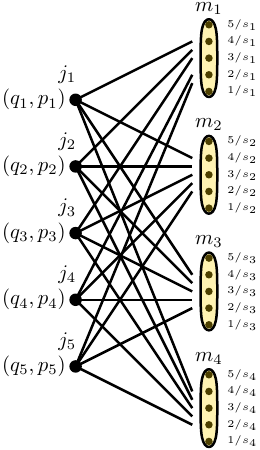}
    \caption{Illustration for the scheduling with uniform machines. The figure shows an example consisting of five jobs and four machines. The left and right sides are job and machine nodes, respectively. Each machine node includes five sub-nodes. Each job node connects to each sub-node of each machine.}
    \label{fig:scheduling-reduction}
\end{figure}

For an instance of multiplicative assignment, we need at least as many suppliers as customers. Therefore, we remove all slots (i.e., customers) except the $m$ slots with the lowest cost, as there is always an optimal solution that does not use these slots.
An optimal solution to the instance of multiplicative
assignment will always use a prefix of slots for each machine,
since the costs are strictly increasing.
It is now straight-forward to transform a solution to the
scheduling problem to one with the same cost in the multiplicative
assignment problem and vice versa.

\subsection{Exact matching and pseudopolynomial time algorithms}
Recall that in the exact bipartite perfect matching problem
we are given a bipartite graph with edges colored either red or blue
and a number $k$.
Our goal is to find a perfect matching that contains exactly $k$ red
edges.
This problem admits a randomized polynomial time algorithm~\cite{mulmuley1987matching}.
More generally, the following problem can be solved in randomized pseudopolynomial time:
given a bipartite graph with edge weights $w: E\rightarrow \N$ and
a target $t\in \N$,
find a perfect matching of weight exactly $t$.
The weighted variant can be reduced to the red-blue variant
by subdividing each edge $2 w(e) - 1$ times 
and coloring the path alternatingly with red and blue~\cite{GurjarKMST16}.
In fact, even the variant with $\ell = O(1)$ many weight 
functions $w_1,w_2,\dotsc,w_{\ell}$ and targets $t_1,t_2,\dotsc,t_{\ell}$ can be solved in randomized pseudopolynomial
time by aggregating the $\ell$ functions into one with
appropriate zero-padding, see e.g.~\cite{Kannan83}. This can be used to derive randomized pseudopolynomial time algorithms for various matching related problems with upgrades.

One can, for example, obtain a
pseudopolynomial time
randomized algorithm for
the multiplicative assignment problem, i.e., whose running time is polynomial
in $n$, $\max_{j\in J} d_j$,
and $\max_{i\in I} c_j$.
Towards this, we use the natural bipartite graph as described in
\Cref{secPreliminiaries}, which contains for each $(i, j)\in I\times j$
an upgraded copy (of cost $b_i \cdot d_j$) and a non-upgraded copy (of cost $c_i \cdot d_j$).
For each $t \in \{1,2,\dotsc, \sum_{j\in J, i\in I} c_i \cdot d_j\}$
we check if there is a matching with cost $t$ that selects
exactly $k$ upgraded edges and output the lowest $t$, for which
there is.
This can be solved using exact bipartite perfect matching
with two weight functions.

We note that this is inferior to the main result of this
paper, which does not require pseudopolynomial time. 
However, the approach via exact matching generalizes to
more complicated problems. For example, it still applies
even if the bipartite graph in multiplicative assignment
is not complete, that is, some suppliers cannot be assigned
to some customers. Furthermore, more complex
cost functions can be implemented, namely, any function
in the supplier, customer,
and whether the supplier is upgraded.
One can even allow both suppliers and customers to be upgraded,
either bounding the total number or---using an additional weight
function---the number of upgraded supplies and customers individually.
Similarly, one may partition the suppliers into a constant number of sets and allow a fixed number of upgrades per set, at the cost
of increasing the running time due to more weight functions. 

Another variant of exact matching is the following optimization
problem, mentioned for example in~\cite{Maalouly23}: given a bipartite graph with red or blue edges, a number
$k$, and a weight function $w: E \rightarrow \N$, find a minimum
weight perfect matching with exactly $k$ red edges.
Again, this problem can be solved in pseudopolynomial time by reduction
to the previous variants.
The problem is not known to be NP-hard for binary encoded weights (which might be exponentially large in the number of suppliers and customers)~\cite{Maalouly23}.
If there was a polynomial time algorithm for it, this would
solve the multiplicative assignment problem, as well as the extensions
mentioned above in polynomial time.
As fascinating as this question is, a solution seems out of reach
since the only known algorithm for exact matching
is via an algebraic framework. It seems illusive to
enhance it such that optimizes an objective function (over the set of
all matchings with exactly $k$ edges), rather than
only computing \emph{some} matching with exactly $k$ red edges.

\subsection{Counter-examples for integrality}
Given our positive results in \Cref{secIntegrality}, one may
wonder if they extend to other cases, specifically those mentioned in the previous subsection, for which there are pseudopolynomial
time algorithms. Since the computational complexity of the optimization
version of exact matching is open (see the previous subsection), 
none of these cases are known to be NP-hard for binary-encoded costs
either.

\paragraph*{Non-complete bipartite graph.}
Let us consider a generalization of the multiplicative assignment problem, where
some customers cannot be assigned to some suppliers, that is, the graph as described in \Cref{secPreliminiaries} is not complete.
For example, in the instance shown in the subfigure (a) of \Cref{fig:noncomplete}, there is no edge from the third left supplier to the second right customer (the suppliers correspond to the depicted nodes on the left and the customers correspond to the nodes on the right).
Suppose we are only allowed to upgrade one supplier (i.e., $k=1$). 
For this case, it is one of the best options to upgrade the top-left supplier, resulting in the optimal matching shown in the subfigure (b) of \Cref{fig:noncomplete}.
The matching has cost
\[
    0 \cdot 3 + 5 \cdot 1 + 2 \cdot 0 = 5.
\]

However, the fractional solution can cheat because the missing edge prevents us from arranging the suppliers increasingly.
Consider the fractional solution shown in the subfigure (c) of \Cref{fig:noncomplete}.
Here, two upgraded edges and four normal edges each receive a value of $\frac{1}{2}$. So, the resulting cost is
\[
    \frac{1}{2} \cdot \left( 0 \cdot 3 + 1 \cdot 1 + 2 \cdot 1 + 5 \cdot 0 + 2 \cdot 3 + 2 \cdot 0 \right) = 
    \frac{9}{2} < 5.
\]

\begin{figure}
    \centering
    \includegraphics[width=0.7\linewidth]{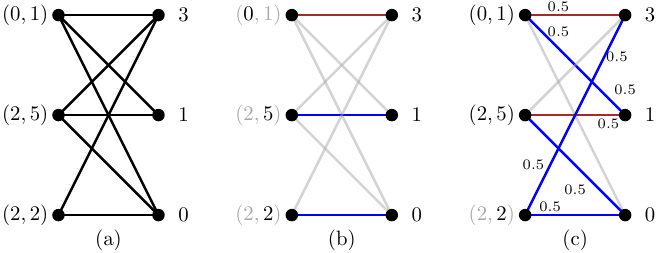}
    \caption{Illustration of non-complete bipartite graph that breaks the integrality. The input bipartite graph is shown in the subfigure (a). The subfigure (b) is the optimal integral solution. Here, we denote upgraded and normal edges in red and blue, respectively. The subfigure (c) is the optimal fractional solution. Here, each blue/red edge will be taken to a fractional extent of $1/2$. It is easy to verify that such a solution is feasible.}
    \label{fig:noncomplete}
\end{figure}

\paragraph*{Partition of suppliers.}
Consider now a partition of the suppliers into subsets with
the constraint that only a specific number can be upgraded per
subset. This is a natural generalization of the cardinality constraint on the set of upgraded suppliers
to a partition matroid. 
We show that the integrality property is lost already for a partition into two sets. 
The instance is shown in the subfigure (a) of \Cref{fig:partition-matroid}: one out of the upper two suppliers and one out of the lower two suppliers are allowed to be upgraded.
\begin{figure}
    \centering
    \includegraphics[width=1\linewidth]{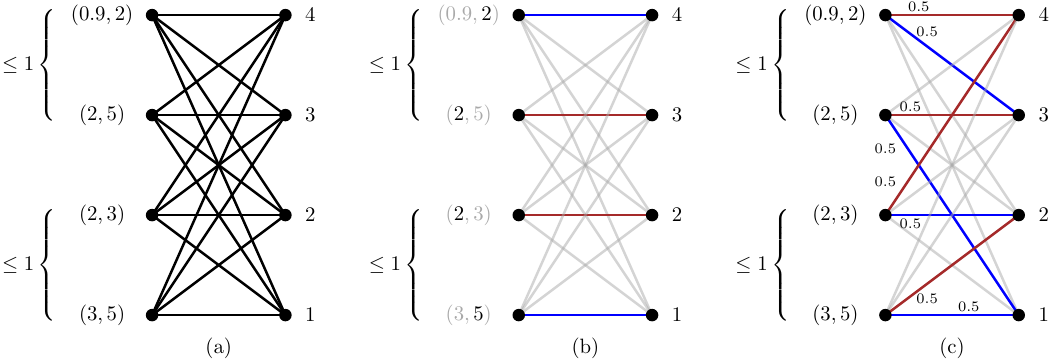}
    \caption{Illustration of a partition matroid that breaks the integrality. There are four suppliers, forming two groups. At most one supplier can be upgraded from each group. The subfigure (b) shows an optimal integral solution, where we upgrade the second node from the first group and the first node from the second group. The blue/red edges represent the normal/upgrading edges, respectively. The subfigure (c) shows an optimal fractional solution where each edge is taken to a fractional extent of $1/2$.}
    \label{fig:partition-matroid}
\end{figure}
\Cref{tab:partition-matroid} shows the cost of upgrading each pair of suppliers, and both the last two solutions are optimal integral solutions.
\begin{table}[htb]
    \centering
    \begin{tabular}{c c}  
    \hline
			upgraded suppliers & cost of the optimal matching \\
			\hline
			$1,3$ & $4\cdot 0.9+3\cdot 2+2\cdot 5+1\cdot 5=24.6$  \\
			$1,4$ & $4\cdot 0.9+3\cdot 3+2\cdot 3+1\cdot 5=23.6$  \\
			$2,3$ & $4\cdot 2+3\cdot 2+2\cdot 2+1\cdot 5=23$  \\
			$2,4$ & $4\cdot 2+3\cdot 2+2\cdot 3+1\cdot 3=23$ \\
			\hline
		\end{tabular}
        \caption{All feasible solutions to the instance shown in the subfigure (a) of \Cref{fig:partition-matroid}.}
    \label{tab:partition-matroid}
\end{table}
Now consider the solution to the LP shown in the subfigure (c) of \Cref{fig:partition-matroid}.
Here each edge is taken to a fractional extent of $1/2$. 
The cost is
\begin{equation*}
    \frac 1 2 \cdot (4\cdot 0.9 + 3 \cdot 2 + 3 \cdot 2 + 1 \cdot 5 + 4 \cdot 2 + 2 \cdot 3 + 2 \cdot 3 + 1 \cdot 5) = 22.8
\end{equation*}
This is lower than the cost of any of the integral solutions.

\paragraph*{Upgrading suppliers and customers.}
Consider now that not only the suppliers can be upgraded, but also the customers (which decreases their demand). 
You are allowed to upgrade at most $k$ suppliers and customers in total. 
The cost of an edge is the (potentially upgraded) demand of the customer multiplied by the (potentially upgraded) cost of the supplier.
This is a natural generalization of the cardinality constraint on upgraded suppliers. We show that we do not have an optimal integral solution to the LP. Consider the instance 
shown in the subfigure (a) of \Cref{fig:dual-upgrading}.

\begin{figure}
    \centering
    \includegraphics[width=0.85\linewidth]{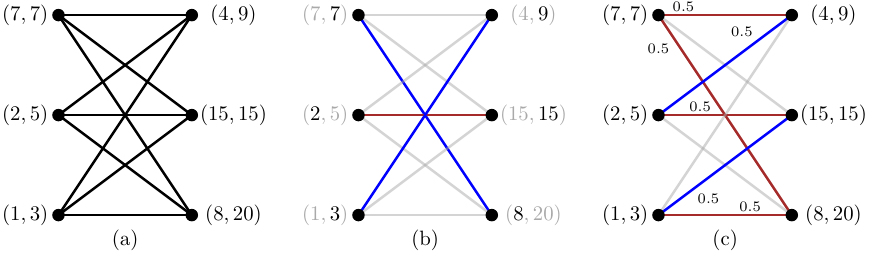}
    \caption{Illustration of upgrading both suppliers and customers that breaks integrality. The subfigure (a) is the input instance. The subfigure (b) is an optimal integral solution. The red/blue edges are upgrading/normal edges. The subfigure (c) is a fractional solution, where each blue/red edge will be taken to the extent of $1/2$.}
    \label{fig:dual-upgrading}
\end{figure}

For $k=1$, we can upgrade the third customer, resulting in an optimal matching of cost $7\cdot 8+5 \cdot 9+3 \cdot 15=146$.
For $k=3$, we can upgrade the second and third supplier and the first customer, resulting in an optimal matching of cost $7\cdot 4+2 \cdot 15+1 \cdot 20=78$.
So by taking the convex combination of the characteristic vectors of these two matchings, we obtain an LP solution for $k=2$ with a cost of $\frac{146+78}{2}= 112$. 
This fractional solution is shown in the subfigure (c) of \Cref{fig:dual-upgrading}.
There is no benefit in upgrading the first supplier or the second customer. For $k=2$, \Cref{tab:dual-upgrading} shows the cost of upgrading each remaining pair.
An optimal integral solution is shown in the subfigure (b) of \Cref{fig:dual-upgrading}.
\begin{table}[htb]
    \centering
\begin{center}
	\begin{tabular}{c c c}
		\hline
		upgraded suppliers & upgraded customers & cost of the optimal matching \\
		\hline
		$2,3$ & none &$7\cdot 9+2\cdot 15+1\cdot 20=113$  \\
		$2$ & $1$    &$7\cdot 4+3\cdot 15+2\cdot 20=113$  \\
		$2$ & $3$    &$7\cdot 8+3\cdot 9+2\cdot 15=113$  \\
		$3$ & $1$    &$7\cdot 4+5\cdot 15+1\cdot 20=123$ \\
		$3$ & $3$    &$7\cdot 8+5\cdot 9+1\cdot 15=116$ \\
		none & $1,3$ &$7\cdot 4+5\cdot 8+3\cdot 15=113$ \\
		\hline
	\end{tabular}
        \caption{All feasible solutions to the instance shown in the subfigure (a) of \Cref{fig:dual-upgrading}. The table does not involve the first supplier and second customer as upgrading them would not change their respective costs).}
    \label{tab:dual-upgrading}
\end{center}
\end{table}
As there is no integral solution of cost at most $112$, there is no optimal integral solution.

\section*{Acknowledgments}
The authors would like to thank Thomas Rothvoss, Laura Sanità, and Robert Weismantel for organizing the 2024 Oberwolfach Workshop on Combinatorial Optimization (2446), where initial results of this work have been discussed.
\bibliographystyle{plain}
\bibliography{main}

\appendix

\section{Greedy algorithm}
\label{apx:greedy}

Consider the following greedy algorithm for the multiplicative assignment problem with upgrades.
We initialize the set $A \coloneqq \emptyset$ and iteratively add a supplier from $I$ to
$A$ in each round, always choosing a supplier that decreases the objective
the most in that round. We stop when $|A|=k$.

This algorithm might fail to compute an optimal solution.
To see this, consider the instance
where $I=J=\{1,2,3\}$ and $(b_{1},c_{1})=(1,5)$, $(b_{2},c_{2})=(0,3)$,
$(b_{3},c_{3})=(3,10)$, $d_{1}=1$, $d_{2}=2$, $d_{3}=3$ (see \Cref{fig:counter_example_greedy}).
In the first iteration, the greedy algorithm selects the supplier
$1$ since $\cost(\{1\})<\cost(\{2\})=\cost(\{3\})$. However, if
$k=2$, then the optimal solution is to upgrade suppliers 2 and 3
since $\cost(\{2,3\})<\cost(\{1,2\})=\cost(\{1,3\})$. Thus, while
upgrading supplier $1$ is optimal for $k=1$, it is not contained
in the unique optimal solution for $k=2$.

\begin{figure}
\centering \includegraphics[width=1\linewidth]{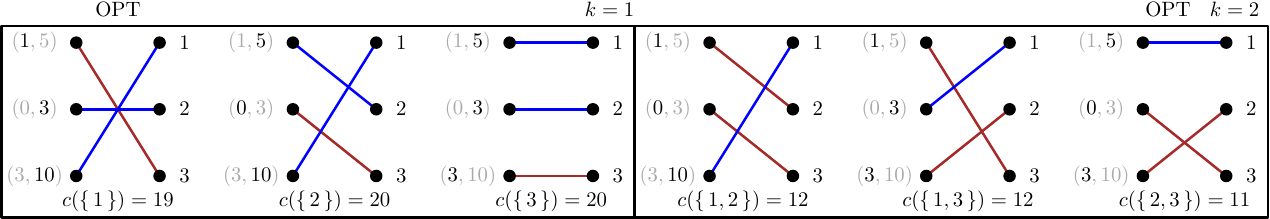}
\caption{ Illustration of an instance in which the greedy algorithm does compute
an optimal solution. Suppliers are depicted on the left, customers
on the right, and upgraded suppliers are marked red. The figure shows
optimal assignments for all possible sets to be upgraded. For $k=1$
the unique optimal solution is on the left, for $k=2$ the unique
optimal solution is on the right. The corresponding upgraded sets
are disjoint.}\label{fig:counter_example_greedy}
\end{figure}

\section{Supermodularity}
\label{sec:supermodularity}
In this section, we will prove that the function $\cost(\cdot)$ is supermodular. We first prove the following auxiliary statement.
\begin{lemma}\label{lem:submodular-aux}
 Consider two instances of the multiplicative assignment problem with upgrades.
 Both instances have the same  suppliers $I = \{1,\dotsc,n\}$,
 customers $J = \{1,\dotsc,n\}$, demands $\bd$, but different 
 costs $\bb, \bc$ and $\bb', \bc'$.
 Further, both cost vectors are ordered non-decreasingly, i.e., $c_{i}\leq c_{i+1}$ and $c'_{i}\leq c'_{i+1}$ for all $i\le n-1$, and the costs in $\bc'$ are at least as large as $\bc$, i.e., $c'_i \ge c_i \ge b_i$ for all $i\in I$.
 Finally, there is a special supplier $s$ with $c_s = c'_{s}$ and $b_s = b'_s$. Then
 \begin{equation*}
     \cost(\emptyset) - \cost(\{s\}) \le \cost'(\emptyset) - \cost'(\{s\}) ,
 \end{equation*}
 where $\cost(\emptyset), \cost'(\emptyset)$ are the optimal objective values for the instances with costs $c, c'$ and no upgraded suppliers, and $\cost(\{s\}), \cost'(\{s\})$ are those after upgrading $c_s, c'_{s}$ to $b_s, b'_s$.
\end{lemma}
\begin{proof}
    Assume without loss of generality that $d_j\geq d_{j+1}$ for all $j\leq n-1$.
    As shown in \Cref{lemComputeAssignmentOnly} the assignment $\pi(j) \coloneqq j$ for $j\in J$ is optimal when not upgrading any suppliers in both instances.
    Let $t\in I$ be minimal with $b_s\leq c_{t}$.
    Then 
    \begin{equation*}
    \pi_s(j) \coloneqq \begin{cases}
		j &\text{ for }j<t \text{ or } j>s \\
		s &\text{ for }j=t \\
		j-1 &\text{ for }t < j\leq s
	\end{cases}
    \quad\text{ and }\quad
        \pi_s^{-1}(i) \coloneqq \begin{cases}
		j &\text{ for }i<t \text{ or } i>s \\
		t &\text{ for }i=s \\
		i+1 &\text{ for }t \le i < s
	\end{cases}
    \end{equation*}
    are an optimal assignment and its inverse for upgrading $\{s\}$ in the instance with costs $\bc$ as shown in the proof of \Cref{lemComputeAssignmentOnly}.
    It follows that
    \begin{align*}
        \cost(\emptyset) - \cost(\{s\}) &= c_s d_s - b d_t + \sum_{i = t}^{s - 1} c_i (d_i - d_{i+1}) \\
        &\le c'_s d_s - b d_t + \sum_{i = t}^{s - 1} c'_i (d_i - d_{i+1}) \\
        &\le \cost'(\emptyset) - \cost'(\{s\}) .
    \end{align*}
    The first inequality holds because all coefficients for $c_i$ (respectively, $c'_i$) are non-negative. The second inequality holds
    because the assignment $\pi_s$ has cost at least $\cost'(\{s\})$ for upgrading $s$ in the instance with costs~$c'$.
\end{proof}
\begin{lemma}
	The set function $\cost(A)$ is supermodular.
\end{lemma}
\begin{proof}
	A set function $\cost$ is supermodular if and only if for all sets $A\subseteq I$ and two distinct elements $s,t\in I \setminus A$ we have
	\begin{equation}\label{eqSubmodularity}
		\cost(A)+\cost(A\cup \{s,t\})\geq \cost(A\cup \{s\})+\cost(A\cup \{t\})
	\end{equation}
    Assume without loss of generality that $c_s \ge c_t$.
    Consider two instances with modified costs $\bc'$ and $\bc''$ that are otherwise identical to the original instance.
    The cost vector $\bc'$ are obtained from upgrading $A\cup \{t\}$ and $\bc''$ are obtained from upgrading $A$.
    Formally,
    \begin{equation*}
        c'_{i} = \begin{cases}
            b_{i} &\text{ if } i\in A\cup \{t\} \\
            c_{i} &\text{ otherwise }
        \end{cases}
        \quad\text{ and }\quad
        c''_{i} = \begin{cases}
            b_{i} &\text{ if } i\in A \\
            c_{i} &\text{ otherwise }
        \end{cases}
    \end{equation*}
    We reorder the components of each of the vectors $\bc'$ and $\bc''$
    non-decreasingly and call the resulting vectors $\bar \bc'$ and $\bar \bc''$, and the corresponding upgraded cost vectors $\bar \bb'$ and $\bar \bb''$. Let $s'$, $t'$ be the new indices of $s$, $t$ in $\bar \bc''$. 
    Since $c_s \ge c_t$ we may break ties in such a way that $s' > t'$.
    Since we can obtain $\bar \bc'$ from $\bar \bc''$ by upgrading $t'$ and reordering the elements $1,2,\dotsc,t'$ with an appropriate tie-breaking rule, we also have that $s'$ is the new index of $s$ in $\bar \bc'$. Let $\cost'(S)$ and $\cost''(S)$ be the optimal objective values in the instances with costs $ \bar \bb', \bar \bc'$ and $\bar \bb'',\bar \bc''$ when $S\subseteq I$ is upgraded. As $c_i'\leq c_i''$ implies $\bar c'_i \le \bar c''_i$ for all $i\in I$, we have by \Cref{lem:submodular-aux}
    that
    \begin{multline*}
        \cost(A) - \cost(A \cup \{s\}) = \cost''(\emptyset) - \cost''(\{s\}) \\
        \ge \cost'(\emptyset) - \cost'(\{s\}) = \cost(A \cup \{t\}) - \cost(A \cup \{s, t\}) . \qedhere
    \end{multline*}
\end{proof}

\section{Deferred proofs}
\label{secDeferredProofs}

\lemComputeAssignmentOnly*
\begin{proof}
    Let $n\coloneqq |J|$. For $i\in I$ let $c_i'\coloneqq \begin{cases}
		b_i &\text{ for }i\in A\\
		c_i &\text{ for }i \in I \setminus A
		\end{cases}$ 
	be the effective cost of the supplier $i$. 
	Let $I=\{1, \dots , |I|\}$ be ordered such that $c_{i}'\leq c_{i+1}'$ for all $i\leq |I|-1$. 
	Also let $J=\{1, \dots , n\}$ be ordered such that $d_j\geq d_{j+1}$ for all $j\leq n-1$.
	Let $\pi(j) \coloneqq j$ for $j\in J$.
	Clearly, this assignment is one-to-one and can be computed in polynomial time. It remains to show that it is optimal.
	
	Let $\pi'$ be an one-to-one assignment that minimizes $\sum_{j\in J}c'_{\pi(j)}d_{j}$. 
	We update $\pi'$ step by step until $\pi'=\pi$ while not increasing the cost of the assignment. 
	Then let $j_1$ be the minimal customer with $\pi'(j_1)\neq j_1$. 
    First suppose that $ \pi'^{-1}(j_1)=\emptyset$. Then we assign customer $j_1$ to supplier $j_1$ instead of supplier $\pi(j_1)$. Formally, let $\pi''(j)\coloneqq \pi'(j)$ for $j\in J\setminus \{j_1\}$ and $\pi''(j_1)\coloneqq j_1$. Then \begin{align*}
		\sum_{j\in J}c'_{\pi''(j)}d_{j}-\sum_{j\in J}c'_{\pi'(j)}d_{j}
		&=(c'_{\pi''(j_1)}-c'_{\pi'(j)})d_{j_1}=(c'_{j_1}-c'_{\pi'(j)})d_{j_1}\leq 0
	\end{align*}
    So the cost of $\pi''$ is at most the cost of $\pi'$.
    
    Now suppose that $\pi'^{-1}(j_1)\neq \emptyset$. Then let $j_2$ the customer assigned to supplier $j_1$. Thus $j_1<j_2$ and $\pi '(j_1)> \pi'(j_2)$. Let $\pi''$ be obtained from $\pi'$ by swapping the values of $j_1$ and $j_2$, i.e. $\pi''(j)\coloneqq \pi'(j)$ for $j\in J\setminus \{j_1, j_2\}$, $\pi''(j_1)\coloneqq \pi'(j_2)$ and $\pi''(j_2)\coloneqq \pi'(j_1)$. Now we have
	\begin{align*}
		\sum_{j\in J}c'_{\pi''(j)}d_{j}-\sum_{j\in J}c'_{\pi'(j)}d_{j}
		&=c'_{\pi''(j_1)}d_{j_1}+c'_{\pi''(j_2)}d_{j_2}-c'_{\pi'(j_1)}d_{j_1}-c'_{\pi'(j_2)}d_{j_2}\\
		&=c'_{\pi'(j_2)}d_{j_1}+c'_{\pi'(j_1)}d_{j_2}-c'_{\pi'(j_1)}d_{j_1}-c'_{\pi'(j_2)}d_{j_2}\\
		&=(c'_{\pi'(j_2)}-c'_{\pi'(j_1)})(d_{j_1}-d_{j_2})\leq 0
	\end{align*}
	as $j_1<j_2$ and $\pi'(j_1)>\pi'(j_2)$ imply $d_{j_1}-d_{j_2}\geq 0$ and $c'_{\pi'(j_2)}-c'_{\pi'(j_1)}\leq 0$. So indeed, the cost of $\pi''$ is at most the cost of $\pi'$. Also the procedure terminates in at most $n-1$ steps with $\pi'=\pi$. Thus $\pi$ is also optimal.
\end{proof}

\lemVertexEdgeMatchings*
\begin{proof}
    For a red edge $e$, let $(\bx, \by)_e:=x_e$ be the LP value corresponding to the edge $e$.
    Similarly, for a blue edge $e$ we set $(\bx, \by)_e:=y_e$.
    Let $C \coloneqq \{e\in E:0<(\bx, \by)_e<1 \}$ be the set of all edges with fractional coordinates in $(\bx, \by)$. 
    The point $(\bx, \by)$ cannot be in the relative interior of a face of $P$ of dimension at least $2$, as this would imply that it lies in the relative interior of a face of dimension at least $1$ when the single Constraint~\eqref{LP4} is added. 
    And this contradicts that $(\bx, \by)$ is a vertex of $P(k)$. 
    So as $(\bx, \by)$ is not a vertex of $P$ by assumption, it lies on an edge of $P$. 
    As shown in e.g.~\cite[Thm.~18.4]{schrijver2003combinatorial}, the set $C$ is a unique cycle. 
    Let $C=\{e_1, \dots, e_{\ell}\}$ such that $e_i$ and $e_{i+1}$ are adjacent for all $i\in [\ell]$ (where we use the notation $e_{\ell+1}\coloneqq e_1$). 
    As the graph is bipartite, $\ell$ is even. 
    Let $A'=\{e_{2i}:i\in [\ell/2]\}$ and $B'=\{e_{2i-1}:i\in [\ell/2]\}$. 
    Due to Constraints \eqref{LP2} and \eqref{LP3} we have $(\bx, \by)_{e_i}+(\bx, \by)_{e_{i+1}}=1$ for all $i\in [\ell]$. 
    This implies $(\bx, \by)_{e_i}=(\bx, \by)_{e_2}$ for $e_i\in A'$ and $(\bx, \by)_{e_i}=(\bx, \by)_{e_1}$ for $e_i\in B'$. 
    Let $M_{\bar{A}}=A'\cup \{e\in E:(\bx, \by)_e=1\}$, $M_{\bar{B}}=B'\cup \{e\in E:(\bx, \by)_e=1\}$ and $\bar{\lambda} \coloneqq (\bx, \by)_{e_2}$. 
    The edge sets $M_{\bar{A}}$ and $M_{\bar{B}}$ are matchings, as no edge $e\in E$ with $(\bx, \by)_e=1$ can be adjacent to a node in $C$.
    This implies $1- \bar{\lambda} = (\bx, \by)_{e_1}$ and $(\bx, \by) =\bar{\lambda} \cdot \chi(M_{\bar{A}})+(1-\bar{\lambda})\chi(M_{\bar{B}})$. 
    
    The sets $\bar{A}$ and $\bar{B}$ are the nodes in $I$ incident to the red edges in $M_{\bar{A}}$ and $M_{\bar{B}}$. If $|\bar{A}|\leq |\bar{B}|$ let $A\coloneqq \bar{A}$, $B\coloneqq \bar{B}$ and $\lambda \coloneqq \bar{\lambda}$. Otherwise let $A\coloneqq \bar{B}$, $B\coloneqq \bar{A}$ and $\lambda \coloneqq 1- \bar{\lambda}$. Then we have $|A|\leq |B|$ and $(\bx, \by) =\lambda \cdot \chi(M_{A})+(1-\lambda)\chi(M_{B})$. Recall that $(\bx, \by)$ is a vertex of $P(k)$. As $\chi(M_{A}), \chi(M_{B})\in P$ we have $\chi(M_{A})\in P \setminus P(k)$ or $\chi(M_{B})\in P \setminus P(k)$. This implies $|A|>k$ or $|B|>k$ which both imply $|B|>k$. As $\lambda |A|+(1-\lambda) |B|\leq k$ and $\lambda\in (0,1)$ we get $|A|<k$. This completes the proof.
\end{proof}

\end{document}